\documentclass[11pt, draftcls, onecolumn]{IEEEtran}
\usepackage{cite}
\usepackage{graphicx}
\usepackage{psfrag}
\usepackage{url}
\usepackage{amsmath}
\usepackage{mathrsfs}
\usepackage{array}
\usepackage{amssymb}
\usepackage{amsfonts}
\usepackage{epstopdf}

\newtheorem{lemma}{Lemma}

\newtheorem{remark}{Remark}
\newtheorem{definition}{Definition}
\newtheorem{theorem}{Theorem}
\newtheorem{example}{Example}

\usepackage{subfigure}

\input{epsf.sty}
\input{bigints.sty}
\newcommand{\SNR}{\text{SNR}}

\title{Interference Aligned Space-Time Transmission with Diversity for the $2 \times 2$ X-Network}
\author{
\authorblockN{Abhinav Ganesan and K. Pavan Srinath}
\thanks{Abhinav Ganesan is with the Institute of Network Coding, The Chinese University of Hong Kong (e-mail: abhinav@inc.cuhk.edu.hk).
K. Pavan Srinath is with the Department of Engineering, University of Cambridge, UK (e-mail: pk423@cam.ac.uk).

A preliminary version of a part of this paper was presented at the 2014 IEEE International Symposium on Information Theory (ISIT). Another part of a content of this paper has been submitted for possible publication to the 2015 IEEE ISIT.}
}

\begin{document}

\maketitle
\thispagestyle{empty}	

\begin{abstract}
The sum degrees of freedom (DoF) of the two-transmitter, two-receiver multiple-input multiple-output (MIMO) X-Network ($2 \times 2$ MIMO X-Network) with $M$ antennas at each node is known to be $\frac{4M}{3}$. Transmission schemes which couple local channel-state-information-at-the-transmitter (CSIT) based precoding with space-time block coding to achieve the sum-DoF of this network are known specifically for $M=2,4$. These schemes have been proven to guarantee a diversity gain of $M$ when a finite-sized input constellation is employed. In this paper, an explicit transmission scheme that achieves the $\frac{4M}{3}$ sum-DoF of the $2 \times 2$ X-Network for arbitrary $M$ is presented. The proposed scheme needs only local CSIT unlike the Jafar-Shamai scheme which requires the availability of global CSIT in order to achieve the $\frac{4M}{3}$ sum-DoF. Further, it is shown analytically that the proposed scheme guarantees a diversity gain of $M+1$ when finite-sized input constellations are employed. 
\end{abstract}	
\begin{keywords}
Interference Alignment, X-Channels, X-Networks, Diversity, Space-time Block Codes, Degrees of Freedom.
\end{keywords}

\section{Introduction} \label{sec1}
The advent of smart phones has led to an explosion in mobile data demand. But a limited spectrum calls for a better spectrum management that incorporates techniques beyond conventional approaches like orthogonalization of spectrum. A further increase in the number of mobile users and data demand means that cell edge users are susceptible to interference from the neighbouring base stations and vice-versa. These issues have instigated research on better transmission techniques in interference networks, with information-theoretic rate tuples often used as the metric for designing better schemes. Since the capacity of interference networks is unknown in general, \textit{degrees of freedom} (DoF) \cite{ZhT} is the widely targeted metric due to its relative ease of characterization. The sum-DoF of a Gaussian network is said to be $d$ if its sum-capacity (in bits per channel use) can be approximated as $C(\SNR)=d \log_2 \SNR + o(\log_2 \SNR)$.

Availability of channel-state-information at the transmitters (CSIT) is an important assumption in the characterization of the approximate capacity of Gaussian interference networks. Availability of perfect global CSIT\footnote{Global CSIT means that all the channel gains of the Gaussian network are available a priori at all the transmitters.} often enables one to design precoders that cast interference onto subspaces independent of the desired signal space at the receivers. This technique, termed interference alignment (IA), was first used implicitly in \cite{MMK_X_Ch_TechRep2008,MMK_X_Ch_TIT2008}, and explicitly appeared in \cite{Jaf_X_Ch_Arxiv2006, JaS_X_Ch_2008} in the context of $2 \times 2$ multiple-input multiple-output (MIMO) X-Networks. A $K\times J$ X-Network is a Gaussian interference network with $K$ transmitters and $J$ receivers and a total of $KJ$ independent messages meant to be sent over the network, one from every transmitter to every receiver. A $2 \times 2$ X-Network with $M$ antennas at each node is referred to as the $(2\times 2, M)$ X-Network. A lower bound on the sum-DoF was shown to be $\lfloor \frac{4M}{3} \rfloor$ for such a network in \cite{MMK_X_Ch_TIT2008}, and it was proven in \cite{JaS_X_Ch_2008} that the sum-DoF equals $\frac{4M}{3}$, achieved using an IA scheme. All the aforementioned works assume the availability of perfect global CSIT. 

The concept of DoF assumes the use of a codebook with unconstrained alphabet size as well as unlimited peak power, but with an average power constraint. The channel is assumed to be static during the transmission of an entire codeword. Further, information-theoretic rate definitions also assume the usage of unlimited coding length. Clearly, all these assumptions are infeasible in practice. In practical communication, the coding length and the codebook size are constrained by factors such as delay requirement and computational complexity. Moreover, the practically used input constellations like QAM and PSK have limited peak power. So, these issues\footnote{In the context of multiuser communication, these issues have motivated the study of the effects of constellation constraints on information-theoretically achievable rates in the two-user multiple access channel \cite{HaR_MAC_TIT2011} and the Gaussian Interference Channel \cite{AbR_2GIC_TWC2012, AbR_KGIC_TWC2014}. However, these works do not take into account limited coding length.} have motivated the research on high reliability communication in MIMO systems under practical constraints like limited coding length, constrained alphabet size, and limited peak power, thus leading to the development of space-time block codes (STBCs) for the single user MIMO systems \cite{TSC}. The theory of STBCs makes the assumption that the channel is constant during the transmission of an entire codeword block but changes independently after every codeword transmission, i.e., the channel is a {\it block fading} one. A metric of significant interest in the design of STBCs is the {\it diversity gain} which indicates the nature of the fall in error probability with SNR. Most of the literature on STBCs is on {\it linear STBCs} \cite{HaH2002} (see Definition \ref{def1} and Definition \ref{def2} in Section \ref{sec2a} for a formal definition of ``STBC'' and ``linear STBC'', respectively) primarily due to the ease of symbol encoding and, to an extent, decoding (using the sphere decoder \cite{ViB}). Associated with such linear STBCs is the notion of symbol rate which is the number of linearly and statistically independent complex symbols transmitted per channel use (see Definition \ref{def3} in Section \ref{sec2a} for a formal definition of ``STBC rate''). It is known that for a single user MIMO system with $M$ transmit antennas and $M$ receive antennas, the maximum possible STBC rate (in complex symbols per channel use) is $M$, which equals the DoF\footnote{For a general $M \times N$ MIMO system, i.e., a MIMO system with $M$ transmit antennas and $N$ receive antennas, the DoF is $\min(M,N)$. For the case where $N < M$, it is currently not known if the best STBC with a rate of $M$ complex symbols per channel use (cspcu) offers any advantage over the best STBC in the comparable class with a rate of $N$ cspcu.} (DoF is the maximum achievable multiplexing gain \cite{ZhT}) of the single user MIMO system.

The above notion of rate (henceforth in this paper, ``rate'' refers to the rate of the STBC unless otherwise mentioned) can be extended to the multiuser setting as follows. Analogous to rate (in a single user MIMO system using STBCs) is the ``sum-rate'' of a linear transmission scheme\footnote{A linear transmission scheme is one where the vectorized version of the symbols received across all the antennas and time instants spanning the codeword length can be expressed as a linear combination of the statistically independent input symbols. In a single user MIMO system, a linear transmission scheme is equivalent to a linear STBC.} in a Gaussian interference network. This sum-rate is a measure of the total number of linearly and statistically independent complex symbols transmitted per channel use (see Definition \ref{def8} in Section \ref{sec2a} for a formal definition of the sum-rate) and is related to the number of independent complex symbols that can be recovered at the receiver by simple zero-forcing. Note that the definition of sum-DoF applies to non-linear transmission schemes while the sum-rate applies strictly to linear schemes with limited coding length and with finite input constellation. However, it is trivially true that the sum-rate cannot exceed the sum-DoF. Therefore, for the $(2\times 2, M)$ X-Network, the maximum sum-rate is $\frac{4M}{3}$ cspcu, achieved by an IA scheme that is linear \cite{JaS_X_Ch_2008}. The primary goal of this paper is to look for linear transmission schemes for the $(2\times 2, M)$ X-Network that achieve the maximum sum-rate along with a non-trivial guaranteed diversity gain when finite and fixed input constellations are employed.

\subsection{Prior Works on Diversity Gain in Interference Networks}
A linear transmission scheme (Definition \ref{def7}, Section \ref{sec2a}) based on the quasi-orthogonal STBC \cite{JafKh_TCOM2001} was proposed for the $(2\times 2)$ X-Network for different configurations of the number of transmit and receive antennas in \cite{LiJ_STBC_X_Ch_TSP2012}. There are several drawbacks with this transmission scheme, though full transmit and receive diversity gains are guaranteed. The transmission scheme requires at least six transmit antennas, and has a sum-rate of 4 cspcu, which does not scale with the number of transmit and receive antennas. Further, the work aims for orthogonality of the desired signals from the two transmitters to a single receiver as well as orthogonality between the desired signal sub-space and the interference sub-space, with the assumption of global CSIT. However, such an orthogonality can easily be achieved without global CSIT using the time division multiple access scheme (TDMA). Another linear transmission scheme achieving an (asymptotic) sum-rate of four cspcu was proposed in \cite{SZX_TCOM2013} for the $(2\times 2)$ X-Network equipped with $M$ transmit and $N$ receive antennas, without the assumption of channel-state-information at any of the transmitters. Clearly, the sum-rate does not scale with the number of transmit or receive antennas, though full transmit and receive diversity gains are guaranteed. Moreover, better sum-rate can be achieved with TDMA along with full transmit and receive diversity gains. Nevertheless, TDMA cannot achieve the maximum sum-rate of $\frac{4M}{3}$ cspcu for the $(2 \times 2,M)$ X-Network.

Linear transmission schemes for the $(2\times 2, 2)$ X-Network and the $(2\times 2, 4)$ X-Network were proposed in \cite{LiJ} and \cite{AbR_4x_X-Ch_TIT2014}. The first linear transmission scheme with a guaranteed diversity gain of 2 with fixed finite input constellations for the $(2\times 2, 2)$ X-Network that achieves the maximum sum-rate of $\frac{8}{3}$ cspcu was proposed in \cite{LJJ,LiJ}. This transmission scheme couples the Alamouti STBC \cite{Ala} with channel-dependent precoding and achieves IA. The same (structure-wise) IA precoding matrices were coupled with the Srinath-Rajan STBC \cite{SrR1} to guarantee a diversity gain of 4 with fixed finite input constellations at the maximum sum-rate of $\frac{16}{3}$ cspcu for the $(2\times 2, 4)$ X-Network \cite{AbR_4x_X-Ch_TIT2014}. In general, STBC designs for single user MIMO systems assume only the availability of perfect channel-state-information at the receivers (CSIR) but not CSIT. However, since the channel matrices are random, CSIT in the $(2\times 2, 2)$ X-Network is inevitable in order to achieve IA, and hence the maximum sum-rate transmission. Moreover, the assumption of CSIT is not an impractical one, since a few state-of-the-art wireless systems support CSIT (for example, the Wi-Fi 802.11ac standard \cite{Aruba_2014}). The precoders of  \cite{LiJ}, which we call the LiJ precoders, assume the availability of local CSIT, i.e., each transmitter is aware of only its own channel matrices to both the receivers, and global CSIR, i.e., all the channel matrices are known to all the receivers. This is in contrast to the assumption of global CSIT (i.e., all the channel matrices are known to all the transmitters) in \cite{JaS_X_Ch_2008} to achieve IA. 

Furthermore, the transmission schemes in \cite{LiJ,AbR_4x_X-Ch_TIT2014} also achieve the $\frac{4M}{3}$ sum-DoF of the $(2\times 2, M)$ X-Network, for $M=2,4$, when the input constellation is Gaussian distributed. In this work, we generalize the above schemes for arbitrary values of $M$. We identify a class of STBCs which when coupled with LiJ precoders achieve the maximum sum-rate (and hence, the sum-DoF\footnote{Throughout the paper, the term ``sum-rate'' pertains to the case where finite input constellations are employed while achievability of ``sum-DoF'' holds relevance when the input constellations are Gaussian distributed.} when utilizing Gaussian distributed input constellations) of $\frac{4M}{3}$ cspcu for the $(2\times 2, M)$ X-Network\footnote{This absence of reduction in the DoF upon the introduction of an STBC is analogous to information-losslessness due to certain STBCs in single-user MIMO systems \cite{HaH2002,SRS2006}.}. The Alamouti STBC and the Srinath-Rajan STBC used in \cite{LiJ,AbR_4x_X-Ch_TIT2014} are special cases of the class we propose in this paper. Moreover, with fixed finite input constellations, a diversity gain of $M+1$ is proven to be guaranteed, and this also establishes that the linear transmission schemes of \cite{LiJ,AbR_4x_X-Ch_TIT2014} achieve a diversity gain of $3$ and $4$ respectively for the $(2\times 2, 2)$ X-Network and the $(2\times 2, 4)$ X-Network . It must be noted that a straightforward generalization of the proof of diversity gain given in \cite{AbR_4x_X-Ch_TIT2014} to the transmission scheme proposed in this paper can guarantee a diversity gain of only $M$. So, the result in this paper on the diversity gain is an improvement over existing ones in the literature.

The contributions of this paper may be summarized as follows. 

\begin{itemize}
\item A class of STBCs, namely STBCs with the column-cancellation property (see Definition \ref{def6} in Section \ref{sec2a}), when coupled with the LiJ precoders is shown to achieve the $\frac{4M}{3}$ sum-DoF of the $(2\times 2, M)$ X-Network. These STBCs are based on STBCs obtained from cyclic division algebras (CDA) \cite{SRS2003}, the explicit construction of which is available in the literature for arbitrary $M$. Since LiJ precoders are used in this work, the $\frac{4M}{3}$ sum-DoF is achieved using local CSIT whereas the Jafar-Shamai scheme \cite{JaS_X_Ch_2008} assumes global CSIT.

\item We prove that when fixed finite input constellations are employed, a diversity gain of $M+1$ is guaranteed with the proposed transmission scheme.

\item For $M=3$, we propose a new STBC with the column-cancellation property and having the minimum possible delay. We show that upon using this STBC in the $(2 \times 2, 3)$ X-Network, the maximum sum-rate of $4$ cspcu and a diversity gain of $4$ with fixed finite input constellations is achieved.
\end{itemize}

The rest of the paper is organized as follows. Section \ref{sec2} provides the signal model and relevant definitions. In Section \ref{sec3}, the proposed linear transmission scheme for the $(2\times 2,M)$ X-Network is presented and it is shown to achieve the maximum sum-rate of $\frac{4M}{3}$ cspcu and also a guaranteed diversity gain of $M+1$ (with fixed finite input constellations) for arbitrary values of $M$. Section \ref{sec4} provides a novel low-delay linear transmission scheme for the $(2\times 2,3)$ X-Network which achieves the maximum sum-rate of $4$ cspcu and a guaranteed diversity gain of 4. The simulation results are presented in Sub-section \ref{subsec4a} and the concluding remarks constitute Section \ref{sec5}. 

\noindent {\it Notation}: Throughout the paper, the following notation is employed. 
\begin{itemize}
 \item Bold, lowercase letters denote vectors, and bold, uppercase letters denote matrices.
 \item $\mathbf{X}^{H}$, $\mathbf{X}^{T}$, $det(\mathbf{X})$, $tr(\mathbf{X})$, $Rank(\mathbf{X})$ and $\Vert \mathbf{X} \Vert$ denote the conjugate transpose, the transpose, the determinant, the trace, the rank, and the Frobenius norm of $\mathbf{X}$, respectively. Further, $\mathbf{X}^*$ denotes the entry-wise conjugation of the elements of $\mathbf{X}$, i.e., $\mathbf{X}^* = \left(\mathbf{X}^{H}\right)^T$.
\item $\textrm{diag}[\mathbf{A}_1,\mathbf{A}_2,\cdots,\mathbf{A}_n]$ denotes a block diagonal matrix with matrices $\mathbf{A}_1$, $\mathbf{A}_2$, $\cdots$, $\mathbf{A}_n$ on its main diagonal blocks.
\item The real and the imaginary parts of a complex-valued vector $\mathbf{x}$ are denoted by $\mathbf{x}_{I}$ and $\mathbf{x}_Q$, respectively. 
\item For a set $\mathcal{S}$, $\vert \mathcal{S}\vert$ denotes its cardinality while for a complex number $x$, $\vert x \vert$ denotes its absolute value.
\item $\mathbf{I}_T$ denotes the identity matrix of size $T \times T$, and $\mathbf{0}$ denotes the null matrix whose dimensions, unless specified in the subscript, are understood from context.
\item For a complex random matrix $\mathbf{X}$, $\mathbb{E}_\mathbf{X}(f(\mathbf{X}))$ denotes the expectation of a real-valued function $f(\mathbf{X})$ over the distribution of $\mathbf{X}$. 
\item $\mathbb{R}$ and $\mathbb{C}$ denote the field of real and complex numbers, respectively. \item Unless used as an index, a subscript or a superscript, $i$ denotes $\sqrt{-1}$. 
\item Unless otherwise specified, for a matrix $\mathbf{X} \in \mathbb{C}^{m\times n}$, $\mathbf{X}(i)$ denotes the $i^{th}$ column of  $\mathbf{X}$, $i \leq N$, and for a set $\mathcal{T} \subset \{1,2,\cdots,N \}$, $\mathbf{X}(\mathcal{T})$ denotes the matrix whose columns are the columns of $\mathbf{X}$ indexed by the elements of $\mathcal{T}$. Further, $\mathbf{X}(i:j,k:l)$ denotes the submatrix of $\mathbf{X}$ consisting of the elements of $\mathbf{X}$ from Row $i$ to Row $j$, Column $k$ to Column $l$, with $1 \leq i < j \leq m$, $1 \leq k < l \leq n $. 
\item For a complex variable $x$, $\check{x}$ is defined as
\begin{equation*}
\check{x} := \left[ \begin{array}{rr}
                             x_I & -x_Q \\
                             x_Q & x_I \\
                            \end{array}\right],
\end{equation*}
and for any matrix $\mathbf{X} \in \mathbb{C}^{n \times m}$, the matrix $\check{\mathbf{X}}$ belonging to $\mathbb{R}^{2n \times 2m}$ is obtained by replacing each entry $x_{ij}$ with $\check{x}_{ij}$, $i=1,2,\cdots, n, j = 1,2,\cdots,m$.
\item The $\widetilde{(.)}$ operator acting on a complex vector is defined as follows. For $\mathbf{x} = [ x_1, x_2, \cdots, x_n ]^T \in \mathbb{C}^{n\times1}$, $\widetilde{\mathbf{x}} := [ x_{1I},x_{1Q}, \cdots, x_{nI}, x_{nQ} ]^T \in \mathbb{R}^{2n\times1}$.
\item $vec(\mathbf{A})$ denotes the vector obtained by stacking the columns of the matrix $\mathbf{A} \in \mathbb{C}^{m\times n}$ one below the other so that $vec(\mathbf{A})  = [\mathbf{A}(1)^T~\mathbf{A}(2)^T\cdots \mathbf{A}(n)^T]^T \in \mathbb{C}^{mn \times 1}$. It follows that, $\widetilde{vec(\mathbf{A})} \in \mathbb{R}^{2mn\times1}$.
\item  The Q-function of $x$ is denoted by $Q(x)$ and given as 
\begin{eqnarray*}
Q(x)  =  \int_{x}^{\infty}\frac{1}{\sqrt{2\pi}}e^{-\frac{t^2}{2}}dt.
\end{eqnarray*}
\item Throughout the paper, $\log x$ denotes the logarithm of $x$ to base 2.
\item The notation $\mathbf{y} \sim {\cal CN}(0,\textbf{I}_T)$ denotes that $\mathbf{y} \in \mathbb{C}^{T \times1}$ has the standard complex normal distribution.
\item $f(x)\doteq x^b$ denotes that $\underset{x \to \infty}{\operatorname{lim}}\frac{\log f(x)}{\log x} = b$, and $\dot{\leq}$ is similarly defined. 
\item $f(x)\doteq g(x)$ denotes that $\underset{x \to \infty}{\operatorname{lim}}\frac{\log f(x)}{\log x} = \underset{x \to \infty}{\operatorname{lim}}\frac{\log g(x)}{\log x}$.
\item $a^+ :=\max(0,a)$.
\item For a real number $a$, $\left\lceil a \right\rceil$ denotes the smallest integer not lower than $a$ while $\left\lfloor a \right\rfloor$ denotes the largest integer not greater than $a$ .
\end{itemize}

\section{Signal Model and Definitions} \label{sec2}
\begin{figure}[htbp]
\centering
\includegraphics[totalheight=3.1in,width=3.5in]{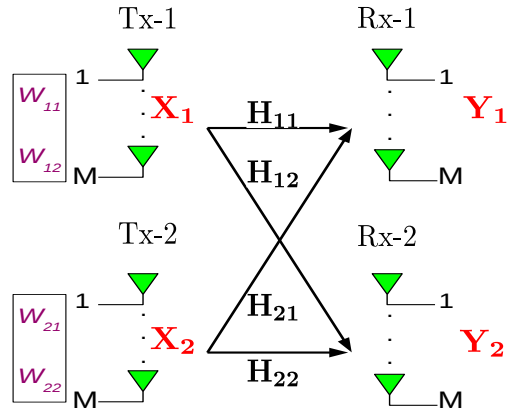}
\vspace{-1cm}
\caption{The $(2\times 2,M)$ X-Network.}
\label{fig-sys_model}
\end{figure}
The $(2\times 2,M)$ X-Network is depicted in Fig. \ref{fig-sys_model}. Two transmitters and two receivers seek to communicate with each other in the presence of additive white Gaussian noise at the receivers. Transmitter $i$ (Tx-$i$) generates an independent message $W_{ij}$ intended for Receiver $j$ (Rx-$j$), $i,j=1,2$. The messages $W_{ij}$ are mapped to a signal matrix $\mathbf{X}_{i} \in \mathbb{C}^{M \times T'}$, $i,j=1,2$. Denoting the output signal matrix at Rx-$j$ by $\mathbf{Y}_{j} \in \mathbb{C}^{M \times T'}$, and the channel matrix from Tx-$i$ to Rx-$j$ by $\mathbf{H}_{ij} \in \mathbb{C}^{M \times M}$, the input-output relation over $T'$ time slots is given by
\begin{align*}
\label{eqn-sys_model}
 \mathbf{Y}_{j}=\sqrt{\rho}\sum_{i=1}^{2} \mathbf{H}_{ij}\mathbf{X}_{i}+\mathbf{N}_j,
\end{align*}where $\mathbf{N}_j \in \mathbb{C}^{M \times T'}$ denotes the noise matrix whose entries are independent and identically distributed (i.i.d.) standard complex normal random variables. The average power constraint at each of the transmitters is $\rho$, and hence $tr\left(\mathbb{E}\left[\mathbf{X}^H_{i}\mathbf{X}_{i}\right]\right)\leq T'$, $i =1,2$. The channel gains are assumed to be constant during the transmission of an entire signal matrix. For the sum-DoF evaluation, the real and imaginary parts of the channel gains are assumed to be distributed independently according to some arbitrary continuous distribution. For the diversity gain evaluation, the channel gains are assumed to be i.i.d. standard complex normal random variables, and experience block-fading. Local CSIT and global CSIR is assumed throughout the paper. 

\subsection{Definitions}\label{sec2a}
A few of the definitions presented below are already available in the literature, while a few other terms are introduced in this paper.
\begin{definition}[Space-Time Block Code \cite{SRS2003}] \label{def1}
 For an $M$ transmit antenna MIMO system, an $(M,T)$ space-time block code (STBC) $\mathcal{X}$ is a finite set of complex matrices of size $M \times {T}$. The block length of the STBC is $T$ channel uses.
\end{definition}

\begin{definition}[Linear STBC \cite{HaH2002}] \label{def2}
An $(M,T)$ STBC $\mathcal{X}$ is called a linear STBC if it can be expressed as 
 \[\mathcal{X} = \left\{ \mathbf{X}=\sum_{i=1}^{k}\mathbf{A}_{iI}x_{iI}+\mathbf{A}_{iQ}x_{iQ} ~ \Big{\vert} ~ \mathbf{A}_{iI}, \mathbf{A}_{iQ}  \in \mathbb{C}^{M \times T}, x_i := x_{iI}+ix_{iQ} \in {\cal Q}_i \right\}, \] 
 where the matrices $\mathbf{A}_{iI},\mathbf{A}_{iQ}$ are called {\it weight matrices} \cite{ZaR}, and $\mathcal{Q}_i$, $i=1,\cdots,k$, are complex constellations with finite cardinality. 
\end{definition}

In the literature, it is generally assumed that $\mathcal{Q}_1 = \mathcal{Q}_2 = \cdots = \mathcal{Q}_k = \mathcal{Q}$ where $\mathcal{Q}$ is either a QAM or a PSK constellation. Linear STBCs are particularly of interest because of the ease of encoding and to an extent, decoding (using the sphere decoder \cite{ViB}). 

\begin{definition}[Rate of a linear STBC]\label{def3}
The rate of an $(M,T)$ linear STBC $\mathcal{X}$ given by 
\[\mathcal{X} = \left\{\mathbf{X}=\sum_{i=1}^{k}\mathbf{A}_{iI}x_{iI}+\mathbf{A}_{iQ}x_{iQ}~ \Big{\vert} ~ \mathbf{A}_{iI}, \mathbf{A}_{iQ}  \in \mathbb{C}^{M \times T}, x_i=x_{iI}+ix_{iQ} \in {\mathcal Q}\right\}\] is said to be $\frac{k}{T}$ complex symbols per channel use (cspcu) if the weight matrices $\mathbf{A}_{iI},\mathbf{A}_{iQ}$ are linearly independent over $\mathbb{R}$.
\end{definition}

Note that rate is not defined to be the number of statistically independent symbols encoded per channel use because an arbitrary number of statistically independent symbols could be packed even in a single dimension. Definition \ref{def3} is inspired by the general design principle that it is more energy-efficient to pack a given number of constellation points in a higher dimensional space than in a lower dimensional space \cite[Chapter $3$]{TsV}. An implication of Definition \ref{def3} is that $\{vec\left(\mathbf{A}_{iI}\right),vec\left(\mathbf{A}_{iQ}\right), i = 1,\cdots,k\}$ is a linearly independent set over $\mathbb{R}$. Associated with every linear STBC is its {\it generator matrix} which is defined as follows. 
\begin{definition}[Generator matrix of a linear STBC \cite{SrR1}]\label{def4}
For an $(M,T)$ linear STBC $\mathcal{X}$ given by 
\[\mathcal{X} = \left\{\mathbf{X}=\sum_{i=1}^{k}\mathbf{A}_{iI}x_{iI}+\mathbf{A}_{iQ}x_{iQ}~ \Big{\vert} ~ \mathbf{A}_{iI}, \mathbf{A}_{iQ}  \in \mathbb{C}^{M \times T}, x_i=x_{iI}+ix_{iQ} \in {\mathcal Q}\right\},\] its generator matrix $\mathbf{G} \in \mathbb{R}^{2MT \times 2k}$ is given by
 \begin{align*}
  \mathbf{G}=\left[\widetilde{vec\left(\mathbf{A}_{1I}\right)} ~\widetilde{vec\left(\mathbf{A}_{1Q}\right)} ~\widetilde{vec\left(\mathbf{A}_{2I}\right)} ~\widetilde{vec\left(\mathbf{A}_{2Q}\right)} ~\cdots~ ~\widetilde{vec\left(\mathbf{A}_{kI}\right)} ~\widetilde{vec\left(\mathbf{A}_{kQ}\right)} \right]
  \end{align*}
so that $\widetilde{vec\left(\mathbf{X}\right)}= \mathbf{G}\widetilde{\mathbf{x}}$ where $\mathbf{x}:=[x_{1} ~x_{2} ~\cdots~ ~x_{k}]^T$. For those linear STBCs of the form \[\mathcal{X} = \left\{\mathbf{X}=\sum_{i=1}^{k}\mathbf{A}_{i}x_{i}~ \Big{\vert} ~ \mathbf{A}_{i} \in \mathbb{C}^{M \times T}, x_i \in {\mathcal Q}\right\},\] we prefer to use the complex version of the generator matrix $\mathbf{G}_C \in \mathbb{c}^{MT \times k}$, which is defined as
\begin{align} \label{eqn_gen_STBC1}
  \mathbf{G}_C=\left[vec\left(\mathbf{A}_{1}\right) ~vec\left(\mathbf{A}_{2}\right)  ~\cdots~ vec\left(\mathbf{A}_{k}\right) \right]
  \end{align}
so that $vec\left(\mathbf{X}\right)= \mathbf{G}_C\mathbf{x}$.
 \end{definition}
 
\begin{definition}[Full-rank STBC \cite{TSC}] \label{def_rank}
An $(M,T)$ STBC $\mathcal{X}$ is said to be full-ranked if
 \[ Rank(\mathbf{X}_1 - \mathbf{X}_2) < M  \Rightarrow \mathbf{X}_1 = \mathbf{X}_2, ~~ \forall \mathbf{X}_1, \mathbf{X}_2 \in \mathcal{X}. \]
\end{definition}

\noindent In other words, full-rankness of an STBC means that the difference matrix of any two distinct codewords of the STBC must be full-ranked.

\begin{definition}[Gaussian-stabilizer function]\label{def5}
 A function $f:\mathbb{C}^{M \times 1}\rightarrow \mathbb{C}^{M \times 1}$ is said to be a Gaussian-stabilizer (GS) function if $f(\mathbf{n}) \sim {\cal CN}\left(\mathbf{0},\mathbf{I}_M\right)$ for $\mathbf{n} \sim {\cal CN}\left(\mathbf{0},\mathbf{I}_M\right)$.
\end{definition}

 Examples of GS-functions are $f(\mathbf{x}) = \mathbf{Ux}$ for any unitary matrix $\mathbf{U} \in \mathbb{C}^{M\times M}$, and $f(\mathbf{x}) = \mathbf{x}^*$. Also, if $f_1$ and $f_2$ are two GS-functions, then so is $f_1 \circ f_2$, where $(f_1 \circ f_2)(\mathbf{x}) := f_1(f_2(\mathbf{x}))$.

\begin{definition}[Column-Cancellation (CC) Property  of an STBC]\label{def6}
Consider an $(M,2T)$ STBC $\mathcal{X}$. Let $\mathcal{T} = \{1, 2,\cdots,2T\}$. Then, $\mathcal{X}$ is said to possess the column-cancellation property if there exist a permutation $\pi : \mathcal{T} \rightarrow \mathcal{T}$ and GS-functions $f_{i},g_{i} :\mathbb{C}^{M \times 1}\rightarrow \mathbb{C}^{M \times 1}$, $i = 1,2,\cdots,T$, such that for every $\mathbf{X} \in \mathcal{X}$,\[\mathbf{X}(\pi(i))+f_{i}\left(\mathbf{X}(\pi(i+T))\right)= g_{i}(\mathbf{X}(\pi(i)))+\mathbf{X}\left(\pi(i+T)\right)=\mathbf{0}, ~~~~~\forall i=1,\cdots,T.\]
\end{definition}
In other words, the CC-property ensures that upon permuting the columns of the codewords of the STBC, the first $T$ columns can be respectively canceled using the last $T$ columns and vice-versa using GS-functions. 
\begin{example}
The $(2,2)$ Alamouti STBC whose codeword matrix is of the form 
\begin{align*}
\mathbf{X}=\begin{bmatrix}
    x_{1} & -x_{2}^*\\
    x_{2} &   x_{1}^*\\
   \end{bmatrix}
\end{align*}
has the CC-property with $T=1$. On choosing $f_1(\mathbf{x}) = \mathbf{P}_1\mathbf{x}^*$, $g_1(\mathbf{x}) = \mathbf{P}_2\mathbf{x}^*$, where
\begin{equation*}
 \mathbf{P}_1 = \left[ \begin{array}{cc}
               0 & -1\\
               1 & 0\\
              \end{array}\right], ~~~ \mathbf{P}_2 = \left[\begin{array}{cc}
               0 & 1\\
               -1 & 0\\
              \end{array}\right],
\end{equation*}
it is clear that the first column of the STBC can be canceled using the second and vice-versa, i.e.,
\begin{align*}
\mathbf{X}(1)+f_1\left(\mathbf{X}(2)\right)=g_1(\mathbf{X}(1))+\mathbf{X}\left(2\right)=\mathbf{0}.
\end{align*}
Note that both $f_1(.)$ and $g_1(.)$ are GS-functions.
\end{example}
\begin{example}
The $(4,4)$ Srinath-Rajan STBC whose codeword matrix is of the form
\begin{align*} \label{eqn-SR_STBC} 
 \mathbf{X}=   \begin{bmatrix}
        x_{1I}+ix_{3Q} & -x_{2I}+ix_{4Q} & e^{i\theta}\left(x_{5I}+ix_{7Q}\right) & e^{i\theta}\left(-x_{6I}+ix_{8Q}\right) \\
	x_{2I}+ix_{4Q} &  x_{1I}-ix_{3Q} & e^{i\theta}\left(x_{6I}+ix_{8Q}\right) & e^{i\theta}\left(x_{5I}-ix_{7Q}\right) \\
	e^{i\theta}\left(x_{7I}+ix_{5Q}\right) & e^{i\theta}\left(-x_{8I}+ix_{6Q}\right) & x_{3I}+ix_{1Q} & -x_{4I}+ix_{2Q} \\
	e^{i\theta}\left(x_{8I}+ix_{6Q}\right) & e^{i\theta}\left(x_{7I}-ix_{5Q}\right) & x_{4I}+ix_{2Q} &  x_{3I}-ix_{1Q}\\
     \end{bmatrix}
\end{align*}for some $\theta \in [0,2\pi)$, also possesses the CC-property with $T=2$. Choosing $\pi(1) =1$, $\pi(2)=3,\pi(3) =2$, $\pi(4)=4$, and GS-functions $f_1(\mathbf{x}) = \mathbf{P}_1\mathbf{x}^*$, $f_2(\mathbf{x}) = \mathbf{P}_2\mathbf{x}^*$, $g_1(\mathbf{x}) = \mathbf{P}_3\mathbf{x}^*$, $g_2(\mathbf{x}) = \mathbf{P}_4\mathbf{x}^*$, where
\begin{equation*}
 \mathbf{P}_1 = \left[ \begin{array}{cccc}
                        0 & -1 &  0 & 0\\
                        1 & 0 & 0 & 0\\
                        0 & 0 & 0 & -e^{2i\theta}\\
                        0 & 0 & e^{2i\theta} & 0 \\
                       \end{array}\right], ~~~ \mathbf{P}_2 = \left[ \begin{array}{cccc}
                        0 & -e^{2i\theta} &  0 & 0\\
                        e^{2i\theta} & 0 & 0 & 0\\
                        0 & 0 & 0 & -1\\
                        0 & 0 & 1 & 0 \\
                       \end{array}\right],
\end{equation*}
\begin{equation*}
 \mathbf{P}_3 = \left[ \begin{array}{cccc}
                        0 & 1 &  0 & 0\\
                        -1 & 0 & 0 & 0\\
                        0 & 0 & 0 & e^{2i\theta}\\
                        0 & 0 & -e^{2i\theta} & 0 \\
                       \end{array}\right], ~~~ \mathbf{P}_4 = \left[ \begin{array}{cccc}
                        0 & e^{2i\theta} &  0 & 0\\
                        -e^{2i\theta} & 0 & 0 & 0\\
                        0 & 0 & 0 & 1\\
                        0 & 0 & -1 & 0 \\
                       \end{array}\right],
\end{equation*}
it is clear that the conditions necessary for the CC-property to hold are satisfied.
\end{example}

\begin{definition}[Linear Transmission Scheme] \label{def7}
Consider a Gaussian interference network\footnote{It must be noted that the terminology ``Gaussian network'', by default, refers to linear channels. A Gaussian interference network has a linear channel with arbitrary (fixed) number of transmitters and an arbitrary (fixed) number of receivers with arbitrary (fixed) message demands.} with $K$ transmitters each having $M$ antennas. Let $\mathbf{X}_i \in \mathbb{C}^{M \times T}$ be the signal matrix that is transmitted over $T$ uses of the channel by Tx-$i$, $i=1,\cdots,K$, with $\mathbf{X}_i = f_i(\mathbf{x}_i)$, where $\mathbf{x}_i \in \mathbb{C}^{k_i \times 1}$ and $f_i : \mathbb{C}^{k_i\times 1} \rightarrow \mathbb{C}^{M \times T}$. Here,  $\mathbf{x}_i$ represents the information bearing symbol vector that Tx-$i$ intends to transmit over the channel and $f_i(.)$ is its encoding function. This transmission scheme $S(K,M,T,f_i(.),k_i)$ is said to be linear if for every $f_i(.)$, $i=1,\cdots,K$, $f_i(a\mathbf{x}_i+a'\mathbf{x}_i') = af_i(\mathbf{x}_i) + a'f_i(\mathbf{x}_i')$, for some complex constants $a$ and $a'$.
\end{definition}

Note that in practice, the symbol vectors $\mathbf{x}_i \in \mathcal{Q}^{k_i \times 1}$ with $\mathcal{Q}$ having finite cardinality. So, it might well be that for $\mathbf{x}_i, \mathbf{x}_i' \in \mathcal{Q}^{k_i \times 1}$, $a\mathbf{x}_i+a' \mathbf{x}_i' \notin  \mathcal{Q}^{k_i \times 1}$, but this has no bearing on Definition \ref{def7}.

\begin{definition}[Sum-rate of a linear transmission scheme]\label{def8}
Consider a Gaussian interference network with $K$ transmitters and $L$ receivers, each having $M$ antennas. For a linear transmission scheme $S(K,M,$ $T, f_i(.),k_i)$, the received signal matrix at Rx-$j$, $j=1,\cdots,L$, is  
\begin{align*}
 \mathbf{Y}_{j}=\sqrt{\rho}\sum_{i=1}^{K} \mathbf{H}_{ij}f_i(\mathbf{x}_i)+\mathbf{N}_j,
\end{align*}
where $\mathbf{N}_j \in \mathbb{C}^{M \times T}$ denotes the noise matrix with its entries being i.i.d. standard complex normal random variables, and $\mathbf{H}_{ij}$ the channel matrix from Tx-$i$ to Rx-$j$ (constant during the transmission of an entire signal matrix). Let $\mathbf{x}_j' \in \mathbb{C}^{k_j' \times 1}$ be the desired symbol vector at Rx-$j$. Then, the sum-rate of $S(K,M,T,f_i(.),k_i)$ is said to be $\frac{\sum_{i=1}^{K}k_i}{T}$ complex symbols per channel use if there exist functions $g_j(.)$ and positive integers $p_j \geq k_j'$, $j=1,\cdots,L$, which satisfy
\begin{eqnarray*}
 g_j : \mathbb{C}^{M \times T} & \longrightarrow & \mathbb{C}^{p_j \times 1} \\
 \mathbf{Y}_j & \longmapsto & \mathbf{A}_j\mathbf{x}_j' + \mathbf{n}
\end{eqnarray*}
where $\mathbf{A}_j \in \mathbb{C}^{p_j \times k_j'}$ which is dependent on $\{\mathbf{H}_{ij}, 1=1,\cdots,K,j=1,\cdots,L\}$ has rank $k_j'$ almost surely, and $\mathbf{n} \sim \mathcal{CN}(\mathbf{0},\mathbf{I}_{p_j})$.
\end{definition}

\begin{remark}
 It is easy to see that the maximum sum-rate (in cspcu) that a linear transmission scheme $S(K,M,T,$ $f_i(.),k_i)$ can achieve equals the sum-DoF of the network. Using standard information-theoretic arguments, it follows that a maximum-sum-rate achieving linear transmission scheme achieves the sum-DoF of the network when the input constellations are Gaussian distributed and the coding length is unlimited. 
\end{remark}

\section{Linear Transmission Scheme for the $(2\times 2,M)$ X-Network} \label{sec3}
We now describe the linear transmission scheme for the general $(2\times 2, M)$ X-Network that achieves the sum-rate of $\frac{4M}{3}$ cspcu. We make use of STBCs from cyclic division algebras (CDA) \cite{SRS2003}. It is well known that STBCs from CDA exist for any number of transmit antennas \cite{ESK_2007}. For a detailed understanding of STBCs from CDA, one can refer to \cite{ESK_2007}, \cite{ORBV}, and references therein. Two key properties of STBCs from CDA that we need in this paper are as follows. Let $\mathcal{X}$ be an STBC from CDA for $M\geq 2$ transmit antennas.  
\begin{enumerate}
 \item For any $\mathbf{X}_1,\mathbf{X}_2 \in \mathcal{X}$, $Rank(\mathbf{X}_1-\mathbf{X}_2) \neq M$ if and only if $\mathbf{X}_1 = \mathbf{X}_2$. In other words, $\mathcal{X}$ is a full-rank STBC (Definition \ref{def_rank}, Section \ref{sec2a}).
 \item $\mathcal{X}$ is an $(M,M)$ linear STBC that encodes $M^2$ linearly and statistically independent complex symbols in $M$ channel uses. Therefore, $\mathcal{X}$ is a rate-$M$ STBC of block length $M$. The complex generator matrix of $\mathcal{X}$, as defined in \eqref{eqn_gen_STBC1}, is of size $M^2 \times M^2$ \cite{ESK_2007}.
\end{enumerate} 

Now, for reasons that are made clear in Theorem \ref{thm_div} and Theorem \ref{thm-max_sum_rate_Mx} that are stated in the following part of this section, we seek full-rank STBCs that have a rate of $M/2$ cspcu and are further equipped with the CC-property. In view of this, we make use of the following lemma.
\begin{lemma}
For every $M \geq 2$, there exist full-rank, rate-$\frac{M}{2}$ STBCs of block length $2T$ for some $T \geq M/2$ that have the CC-property. 
\end{lemma}
\begin{proof}
Let $\mathcal{X}$ be an STBC from CDA. Then, the STBC $\bar{\mathcal{X}}$ given by
\begin{align*}
 \bar{\mathcal{X}} :=  \left\{ [\mathbf{X} ~~\mathbf{PX}] ~ \vert ~ \mathbf{X} \in \mathcal{X} \right\},
\end{align*}
where $\mathbf{P}\in \mathbb{C}^{M \times M}$ is any unitary matrix, has a rate of $\frac{M}{2}$ cspcu and is of block length $2M$. It is easy to check that  $\bar{\mathcal{X}}$ has the CC-property. Since ${\mathcal{X}}$ is full-ranked, so is $\bar{\mathcal{X}}$.  
\end{proof}

\begin{remark} \label{rem1}
It is not necessary that the block length of a full-rank, rate-$M/2$ STBC with the CC-property be at least $2M$. For $M = 2, 4$, we have already shown that the Alamouti STBC and the Srinath-Rajan STBC, which are both full-rank STBCs, have the CC-property and both of them have a rate of $M/2$ cspcu. It turns out that $2\left\lceil \frac{M}{2} \right \rceil$ is the lower bound on the block length of full-rank, rate-$M/2$ STBCs with the CC-property. A general method to construct such minimum-block length STBCs is an open problem.
\end{remark}

Let $\bar{\mathcal{X}}$ denote an $(M,2T)$ STBC equipped with the CC-property\footnote{Henceforth in this paper, it is assumed without loss of generality that the first $T$ columns of the STBC with the CC-property can be canceled using the last $T$ columns and vice-versa. If not, the columns of the STBC can always be permuted to achieve the same.}. The messages $W_{ij}$, with reference to the signal model in Section \ref{sec2}, are mapped to the signal matrices as follows. Each $W_{ij}$ is mapped to $\mathbf{X}_{ij}$ as
\begin{align}\nonumber
 W_{11} \mapsto \mathbf{X}_{11} = \left[\bar{\mathbf{X}}_{11} ~~\mathbf{0}_{M \times T}\right],~~  W_{21} \mapsto  \mathbf{X}_{21} = \left[\bar{\mathbf{X}}_{21} ~~\mathbf{0}_{M \times T}\right],\\
  \label{eqn-STBC_Tx}
  W_{12} \mapsto  \mathbf{X}_{12} = \left[\mathbf{0}_{M \times T} ~~\bar{\mathbf{X}}_{12}\right], ~~  W_{22} \mapsto \mathbf{X}_{22} = \left[\mathbf{0}_{M \times T} ~~\bar{\mathbf{X}}_{22}\right],
\end{align}
where $\bar{\mathbf{X}}_{ij} \in \bar{\mathcal{X}}$. We assume that $\mathbb{E}[\Vert \bar{\mathbf{X}}_{ij}\Vert^2]\leq 2T$ with the codewords being uniformly drawn from $\bar{\mathbf{X}}$. We observe that there is a ``non-zero overlap'' from column $T+1$ to $2T$ between the matrices $\mathbf{X}_{i1}$ and $\mathbf{X}_{i2}$, as also indicated by the hatched regions at the transmitters in Fig. \ref{fig-tx_scheme}. 
\begin{figure}[htbp]
\centering
\includegraphics[height=3.5in,width=6in]{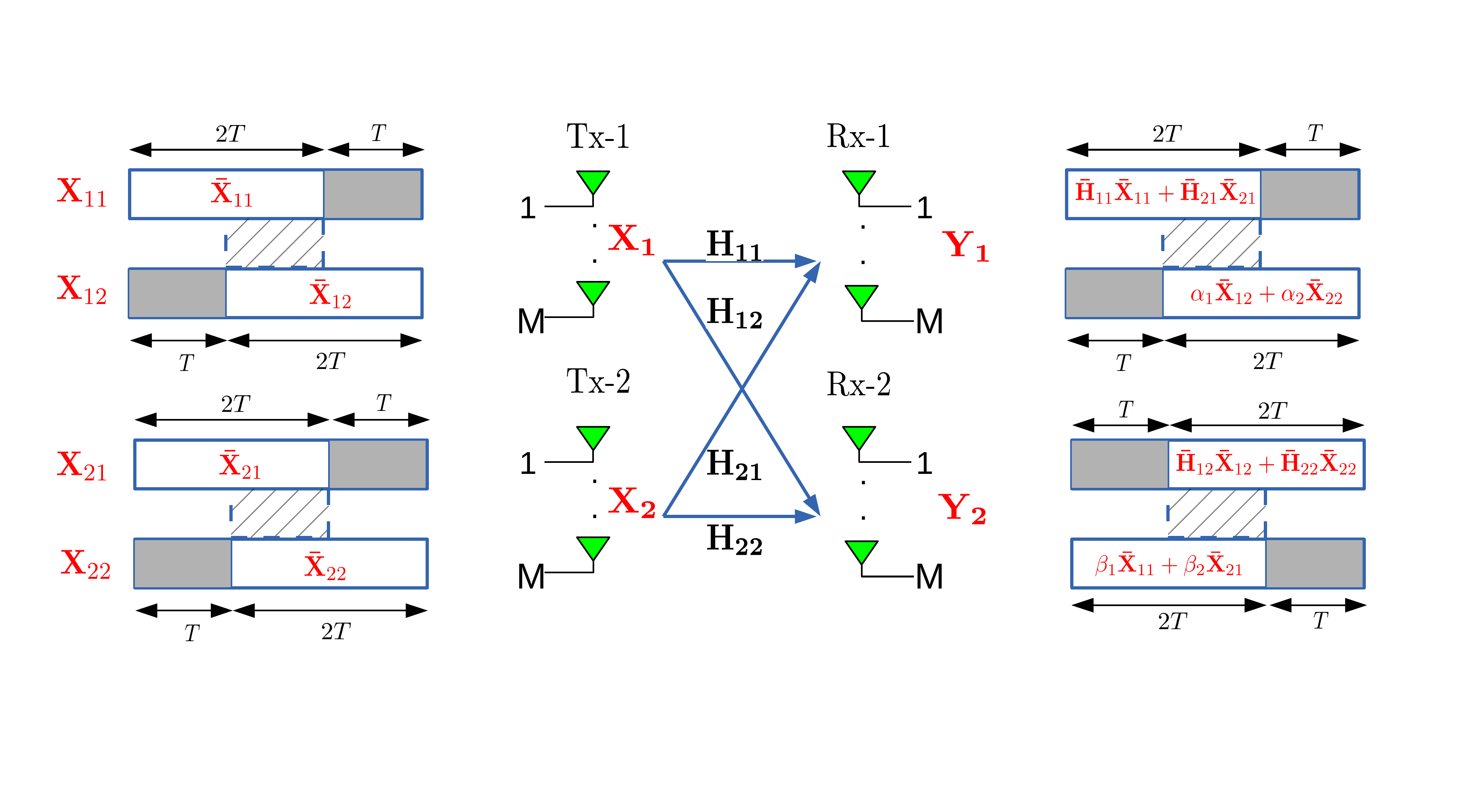}
\vspace{-1cm}
\caption{The transmission scheme that uses STBCs with the CC-property coupled with LiJ precoders is represented here. The power-normalizing scalars involved with the LiJ precoders are denoted by $\alpha_i,\beta_i$, and the effective channel matrices are denoted by $\mathbf{\bar{H}}_{ij}=\mathbf{H}_{ij}\mathbf{V}_{ij}$. The gray shaded regions represent null matrices. The hatched regions at the transmitters indicate the non-zero overlap in the message matrices from time instant $T+1$ to $2T$. The hatched regions at the receivers indicate interference from time instant $T+1$ to $2T$. The interference in the hatched regions is canceled using the CC-property of the STBC used.}
\label{fig-tx_scheme}
\end{figure} 
The transmitted symbols from Tx-$1$ and Tx-$2$ are now (with the average transmit power at each transmitter being limited by $\rho$) given by
\begin{align*}
&\mathbf{X}_1=\sqrt{\frac{3\rho}{4}}\left(\mathbf{V}_{11}\mathbf{X}_{11} + \mathbf{V}_{12} \mathbf{X}_{12}\right),\\
&\mathbf{X}_2=\sqrt{\frac{3\rho}{4}}\left(\mathbf{V}_{21}\mathbf{X}_{21} + \mathbf{V}_{22} \mathbf{X}_{22}\right),
\end{align*}where $\mathbf{V}_{ij}$, $i,j =1,2$, are the LiJ precoders \cite{LiJ} given by 
\begin{align*}
 &\mathbf{V}_{11}:= \frac{\mathbf{H}_{12}^{-1}}{\Vert \mathbf{H}_{12}^{-1} \Vert}, ~\mathbf{V}_{21}:= \frac{\mathbf{H}_{22}^{-1}}{\Vert \mathbf{H}_{22}^{-1} \Vert},\\ & \mathbf{V}_{12}:= \frac{\mathbf{H}_{11}^{-1}}{\Vert \mathbf{H}_{11}^{-1} \Vert},~ \mathbf{V}_{22}:=\frac{\mathbf{H}_{21}^{-1}}{\Vert \mathbf{H}_{21}^{-1} \Vert}.
\end{align*} The LiJ precoders ensure that the effective channel matrices faced by the interference symbols  are scaled identity matrices, and hence are aligned in the same subspace at the unintended receivers. The normalizing factors\footnote{Note that if $\mathbb{E}_{\mathbf{H}}\left[\left \Vert \mathbf{H}^{-1}\right \Vert^2\right] = \mathbb{E}_{\mathbf{H}}\left[tr\left(\left(\mathbf{HH}^H\right)^{-1}\right) \right]$ existed and equalled $a$ (for some positive real number $a < \infty$) for a random matrix $\mathbf{H}$ whose entries are i.i.d. standard complex normal random variables, we could have simply used $1/a$ as the normalizing factor for each $\mathbf{V}_{ij}$. But this unfortunately is not the case \cite{MaK_2000}.} for $\mathbf{H}_{ij}^{-1}$ are chosen to satisfy the power constraint which is $\mathbb{E}_{\mathbf{H}_{i1}}\left[ \left \Vert \mathbf{V}_{i2} \right\Vert^2\right] =\mathbb{E}_{\mathbf{H}_{i2}}\left[ \left \Vert \mathbf{V}_{i1} \right\Vert^2\right] =1$, for $i,j=1,2$.

The received symbol matrix $\mathbf{Y}_1 \in \mathbb{C}^{M \times 3T}$ at Rx-$1$ is given by 
\begin{align}\label{eqn-rx1}
 \mathbf{Y}_1 = \sqrt{\frac{3\rho}{4}} &\left(\mathbf{H}_{11}\mathbf{V}_{11}\mathbf{X}_{11}+\mathbf{H}_{21}\mathbf{V}_{21} \mathbf{X}_{21} + \frac{\mathbf{X}_{12}}{\Vert \mathbf{H}_{11}^{-1} \Vert}+\frac{\mathbf{X}_{22}}{\Vert \mathbf{H}_{21}^{-1} \Vert}\right) + \mathbf{N}_1.
\end{align}It can be observed from the structure of the zero and non-zero columns of $\mathbf{X}_{i2}$ defined in (\ref{eqn-STBC_Tx}) that only the received symbols from time instants $T+1$ to $2T$ face interference, as also indicated by the hatched regions at the receivers in Fig. \ref{fig-tx_scheme}. These interfering symbols can be canceled on account of the CC-property of the STBC used. Define the processed received symbol matrix, obtained after interference cancellation, by
\begin{align*}
&\mathbf{Y}'_1(t):=\mathbf{Y}(t), \text{ for } t=1,\cdots,T,\\
&\mathbf{Y}'_1(t):=\mathbf{Y}(t)+f_{t-T}(\mathbf{Y}(t+T)), \text{ for } t=T+1,\cdots,2T,
\end{align*}
where $f_1(.),f_2(.),\cdots,f_T(.)$ are GS-functions (Definition \ref{def5}, Section \ref{sec2a}). Note that the received symbols from time instants $1$ to $T$ are interference-free because of the presence of zero columns in $\mathbf{X}_{i2}$. We thus have an interference-free processed received symbol matrix $\mathbf{Y}'_1 \in \mathbf{C}^{M\times 2T}$ given by 
\begin{align} \label{eqn-rx1_IF}
&\mathbf{Y}'_1=\sqrt{\frac{3\rho}{4}}\left(\mathbf{H}_{11}\mathbf{V}_{11}\bar{\mathbf{X}}_{11}+\mathbf{H}_{21}\mathbf{V}_{21}\bar{\mathbf{X}}_{21}\right)+\mathbf{N}'_1,
\end{align}where $\mathbf{N}'$ is a noise matrix whose entries are independent but not identically distributed. We have $\mathbf{N}'(i) \sim {\cal CN}(0,\mathbf{I}_M)$, $i=1,2,\cdots, T$, and $\mathbf{N}'(i) \sim {\cal CN}(0,2\mathbf{I}_M)$, $i=T+1,T+2,\cdots, 2T$. Since increasing the noise variance affects neither the achieved DoF nor the diversity gain, we assume that $\mathbf{N}'(i) \sim {\cal CN}(0,2\mathbf{I}_M)$, $i=1,2,\cdots, 2T$.

Similarly, exploiting the CC-property of $\bar{\mathbf{X}}_{i1}$ (where we make use of the GS-functions $g_k(.)$, $k=1,\cdots,T$), the interference-free processed received symbols at Rx-$2$ is given by
\begin{align*}
&\mathbf{Y}'_2=\sqrt{\frac{3\rho}{4}}\left(\mathbf{H}_{12}\mathbf{V}_{12}\bar{\mathbf{X}}_{12}+\mathbf{H}_{22}\mathbf{V}_{22}\bar{\mathbf{X}}_{22}\right)+\mathbf{N}'_2,
\end{align*}where $\mathbf{N}'_2$ has the same distribution as $\mathbf{N}'_1$. Hereafter, we shall focus only on the symbol matrix $\mathbf{Y}'_1$ at Rx-$1$ and any claims about decoding the desired symbols hold good at Rx-$2$ also. Let $P_e$ denote the probability of error in decoding at Rx-$1$. The diversity gain $d_g$ is given by \cite{TSC} 
\begin{equation*}
 d_g = -\lim_{\rho \to \infty}\frac{\log P_e}{\log \rho}.
\end{equation*}

\noindent We now show that a diversity gain of $M+1$ is achievable if the following maximum-likelihood (ML) decoding rule is used.
\begin{align}\label{eqn-ML}
 &(\hat{\mathbf{X}}_{11},\hat{\mathbf{X}}_{21})= \arg \min_{(\bar{\mathbf{X}}_{11},\bar{\mathbf{X}}_{21} \in \bar{\mathcal{X}})}{\left\Vert\mathbf{Y}'_1-\sqrt{\frac{3\rho}{8}}\left(\mathbf{H}_{11}\mathbf{V}_{11}\bar{\mathbf{X}}_{11}+\mathbf{H}_{21}\mathbf{V}_{21}\bar{\mathbf{X}}_{21}\right)\right\Vert}.
\end{align}

It is well-known that a diversity gain of $M^2$ is achieved in a single user $M \times M$ MIMO system with Gaussian distributed channel coefficients when a full-rank STBC is employed \cite{TSC}. Here, we show that when the underlying STBC is full-ranked, a diversity gain of $M+1$ is guaranteed (it goes without saying that the input constellation is of fixed finite cardinality). The loss in the diversity gain relative to the single user MIMO setting is due to the fact that the effective channels seen by the STBCs are not Gaussian  distributed due to channel-dependent precoding at the transmitters. Full receive diversity gain is obtained whereas the transmit diversity gain is affected by precoding.

\begin{theorem} \label{thm_div}
 If the STBC $\bar{\mathbf{X}}$ is full-ranked, then the diversity gain obtained in the $(2\times 2,M)$ X-Network by ML decoding of $(\bar{\mathbf{X}}_{11},\bar{\mathbf{X}}_{21})$ using (\ref{eqn-ML}) is at least $M+1$.
\end{theorem}
\begin{proof}
 The pair-wise codeword error probability that the transmitted codeword pair $(\bar{\mathbf{X}}_{11} , \bar{\mathbf{X}}_{21})$ is erroneously decoded to the codeword pair $(\bar{\mathbf{X}}'_{11} , \bar{\mathbf{X}}'_{21})$, denoted by $P_e(\triangle \bar{\mathbf{X}})$, is given by
\begin{align*} 
\nonumber
P_e(\triangle \bar{\mathbf{X}}) \leq \mathbb{E}_{\mathbf{H}_{ij},i,j=1,2}\left[Q\left(\sqrt{\frac{3\rho}{8}}\left\Vert \mathbf{H}_{11}\mathbf{V}_{11}\triangle \bar{\mathbf{X}}_{11} + \mathbf{H}_{21}\mathbf{V}_{21}\triangle \bar{\mathbf{X}}_{21}\right\Vert\right)\right],
\end{align*}
where $\triangle \bar{\mathbf{X}}_{11}=\bar{\mathbf{X}}'_{11}-\bar{\mathbf{X}}_{11}$ and $\triangle \bar{\mathbf{X}}_{21}=\bar{\mathbf{X}}'_{21}-\bar{\mathbf{X}}_{21}$.  Note that we can have the following three possibilities; 1) $\triangle \bar{\mathbf{X}}_{11} \neq \mathbf{0}$ and $\triangle \bar{\mathbf{X}}_{21}=\mathbf{0}$, 2) $\triangle \bar{\mathbf{X}}_{11} = \mathbf{0}$ and $\triangle \bar{\mathbf{X}}_{21}\neq \mathbf{0}$, 3) $\triangle \bar{\mathbf{X}}_{11} \neq \mathbf{0}$ and $ \triangle \bar{\mathbf{X}}_{21}\neq \mathbf{0}$. We shall prove the statement of the theorem only for the case $\triangle \bar{\mathbf{X}}_{11} \neq \mathbf{0}$, and the proofs for the rest of the cases follow similarly. Let 
\begin{eqnarray}\nonumber
\mathbf{h}^\prime&=&vec\left( \mathbf{H}_{11}\mathbf{V}_{11}\triangle \bar{\mathbf{X}}_{11} + \mathbf{H}_{21}\mathbf{V}_{21}\triangle \bar{\mathbf{X}}_{21}\right)\\
\label{eq_temp}
&=&\left( \triangle \bar{\mathbf{X}}_{11}^T \mathbf{V}_{11}^T \otimes \mathbf{I}_M \right)vec\left(\mathbf{H}_{11}\right)\hspace{-0.1cm}+\hspace{-0.1cm}\left( \triangle \bar{\mathbf{X}}_{21}^T \mathbf{V}_{21}^T \otimes \mathbf{I}_M \right)vec\left(\mathbf{H}_{21}\right).
\end{eqnarray}
Note that \eqref{eq_temp} is due to the simple observation that for $\mathbf{A} \in \mathbb{C}^{m\times n}, \mathbf{B} \in \mathbb{C}^{n \times p}$, it follows that $vec(\mathbf{AB}) = (\mathbf{B}^T \otimes \mathbf{I}_m)vec(\mathbf{A})$. So, we now have
\begin{equation*}\label{eqn-H'}
\Vert\mathbf{H}_{11}\mathbf{V}_{11}\triangle \bar{\mathbf{X}}_{11} + \mathbf{H}_{21}\mathbf{V}_{21}\triangle \bar{\mathbf{X}}_{21}\Vert =\Vert \mathbf{h}^\prime \Vert.
\end{equation*}
Conditioned on the random matrices $\mathbf{H}_{12}$ and $\mathbf{H}_{22}$ which the precoders $\mathbf{V}_{11}$ and $\mathbf{V}_{21}$ respectively depend on, $\mathbf{h}^\prime$ has the same distribution as $\left( \mathbf{K}^{\frac{1}{2}} \otimes\mathbf{I}_M\right)\mathbf{h}$ where \[\mathbf{K}\hspace{-0.1cm}=\hspace{-0.1cm}\left(\triangle \bar{\mathbf{X}}_{11}^T \mathbf{V}_{11}^T\right)\left(\triangle \bar{\mathbf{X}}_{11}^T \mathbf{V}_{11}^T\right)^H + \left( \triangle \bar{\mathbf{X}}_{21}^T \mathbf{V}_{21}^T\right)\left( \triangle \bar{\mathbf{X}}_{21}^T \mathbf{V}_{21}^T\right)^H, \] with $\mathbf{K}=\mathbf{K}^{\frac{1}{2}}{\mathbf{K}^{\frac{1}{2}}}^H $ (since $\mathbf{K}$ is non-negative definite) and $\mathbf{h} \sim \mathcal{CN}(\mathbf{0},\mathbf{I}_{2TM})$.

Using eigen-decomposition\footnote{Any eigen-decomposition that appears in this proof assumes that the eigenvalues are arranged in non-ascending order along the main diagonal of the diagonal eigenvalue matrix.} of $\mathbf{K} \in \mathbb{C}^{2T \times 2T}$ to obtain $\mathbf{K} = \mathbf{U}  \mathbf{\Lambda} \mathbf{U}^H$ with $\mathbf{\Lambda} = \textrm{diag}[\lambda_1,\cdots,\lambda_{2T}]$, we have $\mathbf{K}^{\frac{1}{2}}=\mathbf{U}  \mathbf{\Lambda}^{\frac{1}{2}}\mathbf{U}^H \in \mathbb{C}^{2T \times 2T}$. We now have
\begin{eqnarray*}\nonumber
P_e(\triangle \bar{\mathbf{X}}) & \leq & \mathbb{E}_{\mathbf{H}_{12},\mathbf{H}_{22}}\left[\mathbb{E}_{\mathbf{H}_{11},\mathbf{H}_{21}|{\mathbf{H}_{12},\mathbf{H}_{22}}}\left[Q\left(\sqrt{\frac{3\rho}{8}}\left\Vert \mathbf{h}^\prime \right\Vert\right)\right]\right]\\
 \nonumber
 & = &\mathbb{E}_{\mathbf{H}_{12},\mathbf{H}_{22}}\left[\mathbb{E}_{\mathbf{h}|{\mathbf{H}_{12},\mathbf{H}_{22}}}\left[Q\left(\sqrt{\frac{3\rho}{8}}\left\Vert \left( \mathbf{K}^{\frac{1}{2}} \otimes \mathbf{I}_M \right)\mathbf{h}\right\Vert\right)\right]\right].
 \end{eqnarray*}
 Now, denoting the entries of $\mathbf{h}$ by $ h_i$, $i=1,2,\cdots,2TM$, let $\mathbf{H} \in \mathbb{C}^{M \times 2T}$ be such that $vec(\mathbf{H}) = \mathbf{h}$. Then, 
 \begin{eqnarray*}
 \left\Vert \left( \mathbf{K}^{\frac{1}{2}} \otimes \mathbf{I}_M \right)\mathbf{h} \right\Vert^2 & = & \left\Vert vec\left(\mathbf{H}\left(\mathbf{K}^{\frac{1}{2}}\right)^T\right) \right\Vert^2\\ 
 & = & \left\Vert \mathbf{H}\left(\mathbf{K}^{\frac{1}{2}}\right)^T \right \Vert^2 = \left\Vert \mathbf{K}^{\frac{1}{2}}\mathbf{H}^T \right \Vert^2 = \sum_{i=1}^{M}\left\Vert \mathbf{K}^{\frac{1}{2}}\mathbf{h}_i \right\Vert^2
 \end{eqnarray*}
 where  $\mathbf{h}_i := [h_{i},h_{i+M},h_{i+2M},\cdots,h_{i+(2T-1)M}]^T \in \mathbb{C}^{2T\times1}$, $i=1,\cdots,M$, are the columns of $\mathbf{H}^T$. Therefore,
 \begin{eqnarray*}
 P_e(\triangle \bar{\mathbf{X}}) & = & \mathbb{E}_{\mathbf{H}_{12},\mathbf{H}_{22}}\left[\mathbb{E}_{\mathbf{h}_i, i = 1, \cdots,M|{\mathbf{H}_{12},\mathbf{H}_{22}}}\left[Q\left(\sqrt{\frac{3\rho}{8}\sum_{i=1}^{M}\left\Vert \mathbf{K}^{\frac{1}{2}}\mathbf{h}_i \right\Vert^2}\right)\right]\right] \\
  & \leq & \mathbb{E}_{\mathbf{H}_{12},\mathbf{H}_{22}} \left[\mathbb{E}_{\mathbf{h}_i, i = 1, \cdots,M|{\mathbf{H}_{12},\mathbf{H}_{22}}}\left[Q\left(\sqrt{\frac{3\rho}{8}\sum_{i=1}^{M}\mathbf{h}_i^H {\mathbf{K}^{\frac{1}{2}}}^H \mathbf{K}^{\frac{1}{2}} \mathbf{h}_i }\right)\right]\right]\\
&= &\mathbb{E}_{\mathbf{H}_{12},\mathbf{H}_{22}}\left[\mathbb{E}_{\mathbf{h}_i, i = 1, \cdots,M|{\mathbf{H}_{12},\mathbf{H}_{22}}}\left[Q\left(\sqrt{\frac{3\rho}{8}\sum_{i=1}^{M}\mathbf{h}_i^H \mathbf{U} \mathbf{\Lambda} \mathbf{U}^H  \mathbf{h}_i }\right)\right]\right]
\end{eqnarray*}
where ${\mathbf{K}^{\frac{1}{2}}}^H \mathbf{K}^{\frac{1}{2}} = \mathbf{U} \mathbf{\Lambda} \mathbf{U}^H $. Since $\mathbf{U}$ is unitary, the distribution of $\mathbf{U}^H\mathbf{h}_i$ is the same as that of $\mathbf{h}_i$ so that
\begin{eqnarray*}
P_e(\triangle \bar{\mathbf{X}}) &\leq &\mathbb{E}_{\mathbf{H}_{12},\mathbf{H}_{22}}\left[\mathbb{E}_{\mathbf{h}_i, i = 1, \cdots,M|{\mathbf{H}_{12},\mathbf{H}_{22}}}\left[Q\left(\sqrt{\frac{3\rho}{8}\sum_{i=1}^{M}\left(\mathbf{h}_i^H \mathbf{\Lambda} \mathbf{h}_i \right)}\right)\right]\right]\\
&=&\mathbb{E}_{\mathbf{H}_{12},\mathbf{H}_{22}}\left[\mathbb{E}_{\mathbf{h}_i, i = 1, \cdots,M|{\mathbf{H}_{12},\mathbf{H}_{22}}}\left[Q\left(\sqrt{\frac{3\rho}{8}\sum_{i=1}^{M}\sum_{j=1}^{2T} \lambda_j \vert h_{(j-1)M+i}\vert^2 }\right)\right]\right].
\end{eqnarray*}
Let $\mathbf{K} = \hat{\mathbf{K}} + \tilde{\mathbf{K}}$ where $\hat{\mathbf{K}}:= \left(\triangle \bar{\mathbf{X}}_{11}^T \mathbf{V}_{11}^T\right)\left(\triangle \bar{\mathbf{X}}_{11}^T \mathbf{V}_{11}^T\right)^H $ and $\tilde{\mathbf{K}} := \left( \triangle \bar{\mathbf{X}}_{21}^T \mathbf{V}_{21}^T\right)\left( \triangle \bar{\mathbf{X}}_{21}^T \mathbf{V}_{21}^T\right)^H $. Further, let the eigen-decomposition of $\hat{\mathbf{K}}$ be $\hat{\mathbf{K}} = \hat{\mathbf{U}} \hat{\mathbf{\Lambda}} \hat{\mathbf{U}}^{H}$ and that of $\tilde{\mathbf{K}}$ be $\tilde{\mathbf{K}} = \tilde{\mathbf{U}} \tilde{\mathbf{\Lambda}} \tilde{\mathbf{U}}^{H}$. From Weyl's inequalities\footnote{Weyl's inequalities relate the eigenvalues of the sum of two Hermitian matrices to the eigenvalues of the individual matrices.} (see Section III.2, Page $62$ of \cite{Bha-book}), $\lambda_i \geq \hat{\lambda}_i$, $\forall i = 1,\cdots, 2T$, where $\hat{\mathbf{\Lambda}} = \textrm{diag}[\hat{\lambda}_1,\cdots,\hat{\lambda}_{2T}]$ and $\tilde{\mathbf{\Lambda}}= \textrm{diag}[\tilde{\lambda}_1,\cdots,\tilde{\lambda}_{2T}]$. So, we have
\begin{eqnarray*}
\nonumber
P_e(\triangle \bar{\mathbf{X}}) &\leq & \mathbb{E}_{\mathbf{H}_{12}}
\left[\mathbb{E}_{\mathbf{h}_i, i = 1, \cdots,M|{\mathbf{H}_{12}}}\left[Q\left(\sqrt{\frac{3\rho}{8}\sum_{i=1}^{M}\sum_{j=1}^{2T} \hat{\lambda}_j \vert  h_{(j-1)M+i}\vert^2 }\right)\right]\right]\\
&=&\mathbb{E}_{\mathbf{H}_{12}}\left[\mathbb{E}_{\mathbf{h}_i, i = 1, \cdots,M|{\mathbf{H}_{12}}}\left[Q\left(\sqrt{\frac{3\rho}{8}\sum_{i=1}^{M}\sum_{j=1}^{M} \hat{\lambda}_j \vert  h_{(j-1)M+i}\vert^2 }\right)\right]\right].
\end{eqnarray*}
The last step follows due to the following reason; $\mathbf{V}^T_{11}$ is almost surely invertible and $\triangle \bar{\mathbf{X}}^T_{11}$ is of rank $M$ so that $2T \geq M$. So, $\triangle \bar{\mathbf{X}}^T_{11}\mathbf{V}^T_{11}$ is of rank $M$ and has $M$ non-zero singular values almost surely. Hence, the first $M$ eigenvalues $\hat{\lambda}_1,\cdots,\hat{\lambda}_M$ of $\hat{\mathbf{K}}$ (which are squares of the  singular values of $\triangle \bar{\mathbf{X}}^T_{11}\mathbf{V}^T_{11}$) are non-zero almost surely and the rest of the eigenvalues $\hat{\lambda}_{M+1},\cdots,\hat{\lambda}_{2T}$ are always zero.

Noting that the non-zero eigenvalues of $\hat{\mathbf{K}}$ and  $\hat{\mathbf{K}}':=\left(\triangle \bar{\mathbf{X}}^T_{11} \mathbf{V}^T_{11}\right)^H\left(\triangle \bar{\mathbf{X}}^T_{11} \mathbf{V}^T_{11}\right)$ are the same, and using the fact that multiplication by the unitary eigenvector matrix of $\hat{\mathbf{K}}'$ does not change the distribution of $\mathbf{h}_i' := [h_{i},h_{i+M},h_{i+2M},\cdots,h_{i+(M-1)M}]^T$ (i.e., the first $M$ entries of $\mathbf{h}_i$), we have
\begin{align*}
P_e(\triangle \bar{\mathbf{X}})\leq \mathbb{E}_{\mathbf{H}_{12}}\left[\mathbb{E}_{\mathbf{h}'_i, i = 1, \cdots,M|{\mathbf{H}_{12}}}\left[Q\left(\sqrt{\frac{3\rho}{8}\sum_{i=1}^{M} \mathbf{h}'^H_{i}{\mathbf{V}^T_{11}}^H{(\triangle\bar{\mathbf{X}}_{11}\triangle\bar{\mathbf{X}}^H_{11})}^T\mathbf{V}^T_{11}\mathbf{h}'_i} \right)\right]\right].
\end{align*}Let the eigen-decomposition of $(\triangle\bar{\mathbf{X}}_{11}\triangle\bar{\mathbf{X}}^H_{11})$ be given by $(\triangle\bar{\mathbf{X}}_{11}\triangle\bar{\mathbf{X}}^H_{11})=\mathbf{U}'\mathbf{\Lambda}'\mathbf{U}'^H$, where $\mathbf{\Lambda'}=\text{diag}(\lambda'_1,\cdots,\lambda'_M)$. Multiplication by the unitary matrix $\mathbf{U}'$ does not change the distribution of $\mathbf{V}_{11}$ because $\mathbf{V}_{11}\mathbf{U}'=\frac{{(\mathbf{U}'^H \mathbf{H}_{12})}^{-1}}{\Vert {(\mathbf{U}'^H \mathbf{H}_{12})}^{-1}\Vert}$ and the distribution of $\mathbf{H}_{12}$ is invariant to unitary matrix multiplication. Thus, we have
\begin{align*}
P_e(\triangle \bar{\mathbf{X}})\leq\mathbb{E}_{\mathbf{H}_{12}}\left[\mathbb{E}_{\mathbf{h}'_i, i = 1, \cdots,M|{\mathbf{H}_{12}}}\left[Q\left(\sqrt{\frac{3\rho}{8}\sum_{i=1}^{M} \left(\mathbf{h}'^H_{i}{\mathbf{V}^T_{11}}^H\right)\mathbf{\Lambda'}  \left(\mathbf{V}^T_{11}\mathbf{h}'_i\right)} \right)\right]\right].
\end{align*}
Denoting the least eigenvalue $\lambda'_M$  by $d_{min} >0$ (because $\triangle\bar{\mathbf{X}}_{11}$ is of rank $M$) , we further upper bound the above inequality by
\begin{align*}
P_e(\triangle \bar{\mathbf{X}})\leq\mathbb{E}_{\mathbf{H}_{12}}\left[\mathbb{E}_{\mathbf{h}'_i, i = 1, \cdots,M|{\mathbf{H}_{12}}}\left[Q\left(\sqrt{\frac{3\rho d_{min}}{8}\sum_{i=1}^{M} \mathbf{h}'^H_{i}{\mathbf{V}^T_{11}}^H\mathbf{V}^T_{11}\mathbf{h}'_i} \right)\right]\right].
\end{align*}
Now, using the relations $\mathbf{V}_{11}\mathbf{V}_{11}^H=\mathbf{V} \mathbf{\Lambda}^{(\mathbf{V})} \mathbf{V}^H$ (obtained upon eigen decomposition), $Q(x) \leq \frac{1}{2}e^{\frac{-x^2}{2}}$ for $x \geq 0$, and using the fact the unitary matrix multiplication does not change the distribution of $\mathbf{h}'^H_{i}$, we have
\begin{eqnarray}\nonumber
P_e(\triangle \bar{\mathbf{X}})&\leq&\mathbb{E}_{\mathbf{H}_{12}}\left[\mathbb{E}_{\mathbf{h}'_i, i = 1, \cdots,M|{\mathbf{H}_{12}}}\left[ e^{-\frac{3\rho d_{min}\sum_{i=1}^{M}\sum_{j=1}^{M} {\lambda}^{(\mathbf{V})}_j \vert  h_{(j-1)M+i}\vert^2 }{16}}\right]\right]\\
\label{eq_temp1}
&= & \mathbb{E}_{\mathbf{H}_{12}}\left[\frac{1}{\prod_{j=1}^M\left(1+\frac{3\rho d_{min} {\lambda}^{\mathbf{(V)}}_j }{16}\right)^M}\right].
\end{eqnarray}Note that \eqref{eq_temp1} is arrived at because each of the entries of $\mathbf{h}_i'$ has the complex normal distribution and hence, the square of its absolute value is exponentially distributed with mean 1. The eigenvalues $\lambda^{\mathbf{(V)}}_j$ and the eigenvalues of $\mathbf{H}^H_{12} \mathbf{H}_{12}$ denoted in the non-increasing order by $\lambda_1\geq\cdots\geq \lambda_M \geq 0$ are related as
\begin{align*}
\lambda^{\mathbf{(V)}}_j = \frac{\frac{1}{\lambda_{M-1+j}}}{\sum_{j'=1}^{M}\frac{1}{\lambda_{j'}}} \geq \frac{\lambda_M}{M \lambda_{M+1-j}}.
\end{align*}Hence, by the union bound, the average codeword error probability $P_e$ is upper-bounded as 
\begin{align}\label{eq_div_gain_UB}
 P_e < \vert \bar{\mathcal{X}} \vert^2 \left(\mathbb{E}_{\lambda_j(\mathbf{H}_{12}),j=1,\cdots,M}\left[\frac{1}{\prod_{j=1}^M\left(1+\frac{3\rho d_{min}}{16M}\left(\frac{\lambda_M}{\lambda_j}\right)\right)^M}\right]\right).
\end{align}

Let $\mathbf{H}$ be an $M\times M$ sized matrix with i.i.d. entries $h_{ij} \sim \mathcal{CN}(0,1)$. Denoting the eigenvalues of the complex Wishart matrix $\mathbf{H}^H\mathbf{H}$ by $\lambda_i$, $i=1,\cdots,M$, with $\infty \geq \lambda_1 \geq \cdots \geq \lambda_M \geq 0$, the joint pdf of $\lambda_i$, $i=1,\cdots,M$, is given by
\begin{equation*}
 f_{\boldsymbol{\lambda}}(\lambda_1,\cdots,\lambda_M) = C\prod_{i<j}(\lambda_j -\lambda_i)^2e^{-\sum_{i=1}^M\lambda_i}
\end{equation*}
where $C$ is a normalizing constant. Let $\alpha_i := -\frac{\log \lambda_i}{\log \rho}$. Now, the joint pdf of $\alpha_i$, for $-\infty \leq \alpha_1 \leq \cdots \leq \alpha_M \leq \infty$, is given by
\begin{align*}
 f_{\boldsymbol{\alpha}}(\alpha_1,\cdots,\alpha_M) = C(\log \rho)^M\prod_{i=1}^M \rho^{-\alpha_i}\prod_{i<j}(\rho^{-\alpha_j}-\rho^{-\alpha_i})^2e^{-\sum_{i=1}^M\rho^{-\alpha_i}}.
 \end{align*}

\noindent From (\ref{eq_div_gain_UB}), we have
\begin{align*}
P_e <C\vert \mathcal{X} \vert^2 \bigints_{\boldsymbol{\alpha}} \frac{(\log \rho)^M\prod_{i=1}^M \rho^{-\alpha_i}\prod_{i<j}(\rho^{-\alpha_j}-\rho^{-\alpha_i})^2e^{-\sum_{i=1}^M\rho^{-\alpha_i}}}{\prod_{j=1}^M\left(1+\frac{3\rho d_{min}}{16M}\left(\frac{\rho^{-\alpha_M}}{\rho^{-\alpha_j}}\right)\right)^M}\mathrm{d}\boldsymbol{\alpha}.
\end{align*}
We proceed to analyze the diversity gain achievable using the methodology employed in \cite{ZhT}. Noting that for any $\alpha_i <0$, the integrand has an exponential fall with $\rho$, it is clear that 
\begin{eqnarray}\nonumber
 P_e & \doteq & \bigints_{\boldsymbol{\alpha \geq \mathbf{0}}} \frac{\prod_{i=1}^M \rho^{-\alpha_i}\prod_{i<j}(\rho^{-\alpha_j}-\rho^{-\alpha_i})^2}{\prod_{i=1}^M\left(1+\frac{3\rho d_{min}}{16M}\left(\frac{\rho^{-\alpha_M}}{\rho^{-\alpha_i}}\right)\right)^M}\mathrm{d}\boldsymbol{\alpha}\\\nonumber
 & \leq & \bigints_{\boldsymbol{\alpha \geq \mathbf{0}}} \frac{\prod_{i=1}^M \rho^{-\alpha_i}\prod_{i<j}(\rho^{-\alpha_j}-\rho^{-\alpha_i})^2}{\prod_{i=1}^M\left(\frac{3 d_{min}}{16M}\rho^{(1+ \alpha_i-\alpha_M )^{+}}\right)^M}\mathrm{d}\boldsymbol{\alpha}\\\nonumber
 & \doteq & \bigints_{\boldsymbol{\alpha \geq \mathbf{0}}} \frac{\prod_{i=1}^M \rho^{-\alpha_i}\prod_{i<j}(\rho^{-\alpha_j}-\rho^{-\alpha_i})^2}{\prod_{i=1}^M\left(\rho^{(1+\alpha_M - \alpha_i)^{+}}\right)^M}\mathrm{d}\boldsymbol{\alpha}\\\label{eq_temp2}
 & \leq & \bigints_{\boldsymbol{\alpha \geq \mathbf{0}}} \frac{\prod_{i=1}^M \rho^{-\alpha_i}\prod_{j=2}^M\prod_{i<j}\rho^{-2\alpha_i}}{\prod_{i=1}^{M}\rho^{M(1+\alpha_M - \alpha_i)^{+}}}\mathrm{d}\boldsymbol{\alpha}
\end{eqnarray}
where, with the abuse of notation, $\boldsymbol{\alpha} \geq \mathbf{0}$ implies that $\alpha_i \geq 0, ~ \forall i = 1, \cdots, M$. Note that \eqref{eq_temp2} follows because for $i < j$, we have $\alpha_i \leq \alpha_j$. Therefore,
\begin{eqnarray*}
 P_e & \dot{\leq} & \bigints_{\boldsymbol{\alpha \geq \mathbf{0}}} \prod_{i=1}^M \rho^{-\alpha_i}\prod_{i=1}^M\rho^{-2(M-i)\alpha_i}\prod_{i=1}^M \rho^{-M(1+\alpha_i-\alpha_M)^{+}}\mathrm{d}\boldsymbol{\alpha} \\ 
 & = & \int_{\boldsymbol{\alpha \geq \mathbf{0}}} \rho^{-\sum_{i=1}^M\left(2(M-i)+1\right)\alpha_i +M(1+\alpha_i-\alpha_M)^{+}}\mathrm{d}\boldsymbol{\alpha}\\
 & \doteq & \rho^{-d}
 \end{eqnarray*}
where, from \cite[Theorem 4]{ZhT}, 
\begin{equation*}
 d = \inf_{\boldsymbol{\alpha} \in \mathcal{O}} \sum_{i=1}^M \left(2(M-i)+1\right)\alpha_i +M(1+\alpha_i-\alpha_M)^{+}
\end{equation*}
with $\mathcal{O} = \left\{ \boldsymbol{\alpha} ~ \vert ~ 0 \leq \alpha_1 \leq \alpha_2 \leq \cdots \leq \alpha_M \leq \infty \right\}$. It is easy to verify that the infimum occurs when $\alpha_M =1$ and $\alpha_1 = \alpha_2 = \cdots = \alpha_{M-1} = 0$ so that $d = M+1$. Therefore, we have 
\begin{equation*}
 d_g =-\lim_{\rho \to \infty}\frac{\log P_e}{\log \rho} \geq d = M + 1,
\end{equation*}
which proves the theorem.
\end{proof}

Having shown that the proposed linear transmission scheme guarantees a diversity gain of $M+1$, we now proceed to analyze its sum-rate. Our choice of STBC with the CC-property is the one constructed using STBCs from CDA. Hence, $\bar{\mathbf{X}}_{ij}$, $i,j=1,2$ (with reference to \eqref{eqn-STBC_Tx}), is of the form
\begin{equation}\label{eq_stbc_proposed}
 \bar{\mathbf{X}}_{ij} = \left\{ [\mathbf{R}_{ij} ~~\mathbf{PR}_{ij}] ~ \vert ~ \mathbf{R}_{ij} \in \mathcal{X} \right\}
\end{equation}
where $\mathcal{X}$ is an $(M,M)$ STBC from CDA, and $\mathbf{P}\in \mathbb{C}^{M \times M}$ is a unitary matrix that has {\it no eigenvalue with algebraic multiplicity exceeding $\left\lfloor\frac{M}{2} \right\rfloor$}.

\begin{theorem} \label{thm-max_sum_rate_Mx}
 The proposed linear transmission scheme that uses STBCs from CDA with the unitary matrix $\mathbf{P}$ in \eqref{eq_stbc_proposed} having no eigenvalue with algebraic multiplicity greater than $\left\lfloor\frac{M}{2} \right\rfloor$, has a sum-rate of $\frac{4M}{3}$ cspcu, and hence achieves the sum-DoF of the $(2\times 2,M)$ X-Network.
\end{theorem}
\begin{proof}
Let $\mathcal{T}_1 := \{1,2,\cdots,M\}, \mathcal{T}_2 := \{M+1,2,\cdots,2M\}, \mathcal{T}_3 := \{2M+1,2M+2,\cdots,3M\}$. The received symbol matrix at Rx-$1$ in (\ref{eqn-rx1}) can be represented as
\begin{align*}
		 \left[\begin{array}{c}
                       \mathbf{Y}_1(\mathcal{T}_1)\\
                       \mathbf{Y}_1(\mathcal{T}_2)\\
                       \mathbf{Y}_1(\mathcal{T}_3)\\
                      \end{array}\right] =  \sqrt{\frac{3\rho}{4}}\left[\begin{array}{cccc}
                      \mathbf{H}_{11}\mathbf{V}_{11} & \mathbf{H}_{21}\mathbf{V}_{21} & \mathbf{0} & \mathbf{0} \\
                      \mathbf{H}_{11}\mathbf{V}_{11}\mathbf{P} & \mathbf{H}_{21}\mathbf{V}_{21}\mathbf{P} &\alpha_1 \mathbf{I}_M &\alpha_2 \mathbf{I}_M\\
                      \mathbf{0} & \mathbf{0} & \alpha_1\mathbf{P} & \alpha_2\mathbf{P}\\
                      \end{array}                      
                      \right]\left[\begin{array}{c}
                       \mathbf{R}_{11}\\
                        \mathbf{R}_{21}\\
                        \mathbf{R}_{12}\\ 
                         \mathbf{R}_{22}\\
                      \end{array}\right]+ \left[\begin{array}{c}
                       \mathbf{N}_1(\mathcal{T}_1)\\
                       \mathbf{N}_1(\mathcal{T}_2)\\
                       \mathbf{N}_1(\mathcal{T}_3)\\
                      \end{array}\right],
\end{align*}
where $\alpha_1 := \frac{1}{\Vert \mathbf{H}_{11}^{-1} \Vert}$ and $\alpha_2 := \frac{1}{\Vert \mathbf{H}_{21}^{-1} \Vert}$. Now, the processed interference-free received symbol matrix $\mathbf{Y}'_1$ is given by
\begin{eqnarray} \nonumber
\mathbf{Y}'_1	& = &	 \left[\begin{array}{c}
                       \mathbf{Y}'_1(\mathcal{T}_1)\\
                       \mathbf{Y}'_1(\mathcal{T}_2)\\
                      \end{array}\right] = \mathbf{F}\left[\begin{array}{c}
                       \mathbf{Y}_1(\mathcal{T}_1)\\
                       \mathbf{Y}_1(\mathcal{T}_2)\\
                       \mathbf{Y}_1(\mathcal{T}_3)\\
                      \end{array}\right]\\ \label{eqn-eff_noise}
                      & = &\sqrt{\frac{3\rho}{4}}\underbrace{\left[\begin{array}{cc}
                           \mathbf{H}_{11}\mathbf{V}_{11} & \mathbf{H}_{21}\mathbf{V}_{21}  \\ 
\mathbf{H}_{11}\mathbf{V}_{11}\mathbf{P} & \mathbf{H}_{21}\mathbf{V}_{21}\mathbf{P}\\
\end{array}\right]}_{\mathbf{H}_1}\left[\begin{array}{c}
                       \mathbf{R}_{11}\\
                        \mathbf{R}_{21}\\                        
                      \end{array}\right]+\left[\begin{array}{c}
                       \mathbf{N}_1(\mathcal{T}_1)\\
                       \mathbf{N}_1(\mathcal{T}_2)-\mathbf{P}^H\mathbf{N}_1(\mathcal{T}_3)\\                       
                      \end{array}\right],
\end{eqnarray}
where the interference zero-forcing matrix $\mathbf{F}$ is given by
\begin{equation*}
 \mathbf{F}= \left[\begin{array}{ccr}
                      \mathbf{I}_M & \mathbf{0} & \mathbf{0} \\
                      \mathbf{0} & \mathbf{I}_M & -\mathbf{P}^H\\
                      \end{array}\right].
\end{equation*}
Adding a Gaussian noise matrix to \eqref{eqn-eff_noise} (which one can note doesn't affect the claim about sum-rate), we now assume that the entries of the effective noise matrix in (\ref{eqn-eff_noise}) are i.i.d. standard complex normal random variables.

Since $\mathbf{R}_{11}$ and $\mathbf{R}_{21}$ are codewords of the same $(M,M)$ STBC $\mathcal{X}$ which is obtained from CDA, $vec(\mathbf{R}_{11}) = \mathbf{G}_C\mathbf{x}_{11}$, $vec(\mathbf{R}_{21}) = \mathbf{G}_C\mathbf{x}_{21}$ where $\mathbf{G}_C \in \mathbb{C}^{M^2 \times M^2}$ is the complex generator matrix of $\mathcal{X}$, and $\mathbf{x}_{11}$, $\mathbf{x}_{21} \in \mathbb{C}^{M^2 \times1}$ are the complex information symbol vectors that are meant to be decoded at Rx-$1$. Let $\mathbf{G}_C = [ \mathbf{G}_{1}^T ~ \mathbf{G}_2^T ~\hdots\mathbf{G}_M^T]^T$, where $\mathbf{G}_i \in \mathbb{C}^{M\times M^2}$. Then, we have
\begin{align*}
 {y}_1 	:= vec\left(\frac{\mathbf{Y}'_1}{\sqrt{2}}\right)=\sqrt{\frac{3\rho}{8}}\underbrace{\left(\mathbf{I}_{M} \otimes \mathbf{H}_1\right)
 \left[\begin{array}{ll}
   \mathbf{G}_1&\mathbf{0}\\
   \mathbf{0} & \mathbf{G}_1\\
   \mathbf{G}_2&\mathbf{0}\\
   \mathbf{0} & \mathbf{G}_2\\
   \vdots & \vdots \\
   \mathbf{G}_M &\mathbf{0}\\
   \mathbf{0} & \mathbf{G}_M\\
  \end{array}\right]}_{\mathbf{H}_{eq_1}}\left[\begin{array}{l}
  \mathbf{x}_{11} \\ \mathbf{x}_{21}\\ \end{array}\right]
  + \mathbf{n}'_1,
\end{align*}
where $\mathbf{n}'_1 \sim {\cal CN}\left(\mathbf{0}, \mathbf{I}_{2M^2}\right)$. Since $\mathbf{G}_C$ is full-ranked (Definition \ref{def4}, Section \ref{sec2a}), the equivalent channel matrix $\mathbf{H}_{eq_1} \in \mathbb{C}^{2M^2 \times 2M^2}$ is full-ranked with probability 1 if and only if $\mathbf{H}_1$ is. To show that $\mathbf{H}_1$ is full-rank with probability 1 for our choice of the unitary matrix $\mathbf{P}$, we make use of the following result.

\begin{lemma}[\cite{caron}]\label{lemma1}
 A polynomial function on $\mathbb{R}^n$ to $\mathbb{R}$ is either identically 0, or non-zero almost everywhere.
\end{lemma}

Lemma \ref{lemma1} holds when $\mathbb{R}$ is replaced by $\mathbb{C}$ with its proof being on the same lines as that in \cite{caron}. Note that $det(\mathbf{H}_1)$ is a polynomial function on $\mathbb{C}^{4M^2}$ to $\mathbb{C}$, the variables being the $M^2$ entries of each\footnote{Though the channel-dependent power-normalizing factors are present in the denominator of the terms of $det(\mathbf{H}_1)$, they can be ignored. This is because these factors are almost surely non-zero and the columns of $\mathbf{H}_1$ can be de-normalized without affecting the rank of $\mathbf{H}_1$.} of $\mathbf{H}_{11}$, $\mathbf{H}_{12}$, $\mathbf{H}_{21}$ and $\mathbf{H}_{22}$. We now show that $det(\mathbf{H}_1)$ is not identically 0. Let $\mathbf{H}_{11}\mathbf{V}_{11} = \mathbf{I}_M$ and let $\mathbf{A} := \mathbf{H}_{21}\mathbf{V}_{21}$. So, 
\begin{equation*}
 det(\mathbf{H}_1) = det\left(\left[\begin{array}{ll}
                                     \mathbf{I}_M & \mathbf{A} \\
                                     \mathbf{P} & \mathbf{AP}\\
                                    \end{array}\right] \right)
\end{equation*}
which is 0 iff $\mathbf{AP}-\mathbf{PA}$ is singular. The following result now proves that $\mathbf{H}_1$ is full-ranked with probability 1.
\begin{lemma}\label{lemma2}
For a unitary matrix $\mathbf{P}$, there exists $\mathbf{A} \in \mathbb{C}^{M \times M}$ such that $\mathbf{AP}-\mathbf{PA}$ is full-ranked if and only if $\mathbf{P}$ has no eigenvalue with algebraic multiplicity greater than $\left\lfloor\frac{M}{2} \right\rfloor$.
\end{lemma}

The proof of Lemma \ref{lemma2} has been provided in Appendix \ref{appen_lem_full_rank}. Lemma \ref{lemma2} establishes that $\mathbf{H}_{eq_1}$ is full-ranked with probability 1, and hence Rx-$1$ gets $2M^2$ linearly independent complex symbols in $3M$ channel uses. An analysis on similar lines reveals that Rx-$2$ also obtains $2M^2$ linearly independent complex symbols in $3M$ channel uses. Therefore, the sum-rate of the proposed linear transmission scheme is $\frac{4M}{3}$ cspcu. Thus, the transmission scheme achieves the $\frac{4M}{3}$ sum-DoF of the $(2\times 2,M)$ X-Network upon using Gaussian distributed input constellations in which case the notion of diversity gain is no longer relevant.
\end{proof}

\section{Full-rank, Minimum Delay, rate-$\frac{M}{2}$ STBC with the CC-property for $M=3$} \label{sec4}
Our proposed linear transmission scheme that achieves the maximum sum-rate for the $(2\times 2,M)$ X-Network made use of STBCs from CDA. The STBCs with the CC-property that we have constructed so far have a rate of $M/2$ cspcu and a block-length of $2M$. As pointed out in Remark \ref{rem1}, it is not necessary for rate-$\frac{M}{2}$ STBCs with the CC-property to have a block length of at least $2M$. There are two significant advantages in employing rate-$\frac{M}{2}$ STBCs with the CC-property with a block length less than $2M$. 
\begin{enumerate}
 \item Firstly, a rate-$\frac{M}{2}$ STBC of block length $2T$ encodes $MT/2$ complex symbols. This means that at each receiver in the $(2\times 2,M)$ X-Network, a joint decoding of $MT$ complex symbols needs to be performed. So, it would be advantageous to have $T$ as small as possible. The lower bound on $T$ is $\left\lceil \frac{M}{2}\right \rceil$, which follows from the full-rankness condition required in Theorem \ref{thm_div}. The Alamouti STBC and the Srinath-Rajan STBC achieve this lower bound on $T$.
 \item At each receiver, the decoder needs to wait for $3T$ channel uses before it can proceed with the decoding. In view of tight delay-requirements, a low decoding-delay (which is the number of time-slots that the decoder has to wait for before proceeding to decode the symbols) is desirable. 
\end{enumerate}

Since the notion of STBCs with the CC-property is introduced in this paper, the problem of designing minimum-delay, full-rank, rate-$\frac{M}{2}$ STBCs with the CC-property for arbitrary $M$ is open. Such STBCs are known only for $M =2$ (the Alamouti STBC) and $M=4$ (the Srinath-Rajan STBC). In the following part of this section, we propose a full-rank, rate-${\frac{3}{2}}$ STBC with the CC-property having the minimum value of 4 for $T$. Using this STBC would incur a decoding delay of $6$ channel uses for the $(2\times 2,3)$ X-Network. We further show that the linear transmission scheme using this STBC achieves the maximum sum-rate of the $(2\times 2,3)$ X-Network.

The STBC $\bar{\mathcal{X}}$ with the CC-property for $M=3$ has its codewords $\bar{\mathbf{X}}$ of the form  
\begin{align} \label{eqn-M=3_STBC}
   \bar{\mathbf{X}}= \begin{bmatrix}
     s_1 & e^{i\theta}s_4 & -s_2^* &   -e^{i\theta}s_6^* \\
	s_2 & e^{i\theta}s_5 & s_1^* &  e^{i\theta}s_4^* \\
	e^{i\theta}s_3 & s_6 & -e^{i\theta}s_3^*& -s_5^* \\
     \end{bmatrix}
\end{align}

\noindent where $s_1 = x_{1I}+ix_{3Q}$, $s_2 = x_{2I}+ix_{4Q}$, $s_3 = x_{6I}+ix_{5Q}$, $s_4 = x_{5I}+ix_{6Q}$, $s_5 = x_{4I}+ix_{2Q} $, $s_6 = x_{3I}+ix_{1Q}$, and $\theta \in [0,2\pi)$. Note that the actual complex symbols $x_i$, $i=1,2,\cdots,6$, take values independently from a complex constellation $\mathcal{Q}$. In order to identify conditions on $\theta$ and $\mathcal{Q}$ that that need to be satisfied for $\bar{\mathcal{X}}$ to have full-rank, we make use of the following definition.

\begin{definition}[Coordinate Product Distance of a complex constellation \cite{ZaR}]
The Coordinate Product Distance (CPD) of a complex constellation $\mathcal{Q}$ is defined as
\begin{align*}
 CPD(\mathcal{Q})= \min_{u,v\in\mathcal{Q}, u \neq v}\left|u_I-v_I\right|\left|u_Q-v_Q\right|.
\end{align*}
\end{definition}

If a constellation has a CPD of zero, it can be rotated appropriately so that the resulting constellation has a non-zero CPD \cite{ZaR}. It must be observed that 
the product $\left|u_I-v_I\right|\left|u_Q-v_Q\right|$ is equal to zero for a constellation with non-zero CPD if and only if $u=v$.
\begin{lemma} \label{lem_diff_mat_FR}
There exists $\theta \in [0,2\pi)$ such that when $x_i$, $i=1,\cdots,6$, take values from a complex constellation $\mathcal{Q}$ with a non-zero CPD, $\bar{\mathcal{X}}$ is full-ranked.
\end{lemma}
\begin{proof}
 The proof has been provided in Appendix \ref{appen_lem_diff_mat_FR}.
\end{proof}

Note that the STBC $\bar{\mathcal{X}}$ has the CC-property for the choice of GS-functions $f_1(\mathbf{x}) = \mathbf{P}_1\mathbf{x}^*$, $f_2(\mathbf{x}) = \mathbf{P}_2\mathbf{x}^*$, $g_1(\mathbf{x}) = \mathbf{P}_3\mathbf{x}^*$, $g_2(\mathbf{x}) = \mathbf{P}_4\mathbf{x}^*$, where
\begin{align*}
 &\mathbf{P}_1 = \left[ \begin{array}{ccc}
                        0 & -1 &  0 \\
                        1 & 0 & 0 \\
                        0 & 0 & e^{2i\theta}\\
                        \end{array}\right], ~~~ \mathbf{P}_2 = \left[ \begin{array}{ccc}
                        0 & -e^{2i\theta} &  0\\
                        0 & 0 & e^{i\theta} \\                       
                        e^{i\theta} & 0 & 0 \\
                       \end{array}\right],\\
&\mathbf{P}_3 = \left[ \begin{array}{ccc}
                        0 & 1 &  0\\
                        -1 & 0 &  0\\
                        0 & 0 &  e^{2i\theta}\\                      
                       \end{array}\right], ~~~ \mathbf{P}_4 = \left[ \begin{array}{ccc}
                        0 & 0 & e^{i\theta}\\
                        -e^{2i\theta} & 0 & 0\\
                        0 & e^{i\theta} & 0\\                       
                       \end{array}\right].
\end{align*}
We now prove that using a linear transmission scheme based on $\bar{\mathcal{X}}$ achieves the maximum sum-rate of 4 cspcu for the $(2\times 2,3)$ X-Network.

\begin{theorem} \label{thm_x3_DoF}
The proposed linear transmission scheme based on $\bar{\mathcal{X}}$ achieves the maximum sum-rate of 4 cspcu for the $(2\times 2,3)$ X-Network for any $\theta \in [0,2\pi)$.
\end{theorem}
\begin{proof}
The proof has been provided in Appendix \ref{appen_thm_x3_DoF}.
\end{proof}

\subsection{Simulation Results}\label{subsec4a}
We consider the $(2\times 2, 3)$ X-Network and plot the bit error rates (BER) for two linear transmission schemes; the low-delay transmission scheme based on the rate-$\frac{3}{2}$ STBC whose codewords are of the form given in \eqref{eqn-M=3_STBC}, and the transmission scheme based on the perfect STBC for 3 antennas \cite{ORBV} whose codewords $\bar{\mathbf{X}}$ are of the form shown below. 
\begin{align*}
\bar{\mathbf{X}} = \left[\begin{array}{ccc}
                    s_1 & \omega s_8 & \omega s_6\\
                    s_4 & s_2 & \omega s_9 \\
                    s_7 & s_5 & s_3 \\
                   \end{array} \right]
\end{align*}
where $\omega = e^{\frac{2i\pi}{3}}$ and for $j = 0,1,2$,
\begin{equation*}
 \left[\begin{array}{c}
        s_{3j+1} \\ s_{3j+2} \\ s_{3j+3} \\
       \end{array}
 \right]  = \left[ \begin{array}{ccc}
                    0.6603 + 0.3273i &  0.0207 + 0.3273i & -0.4920 + 0.3273i\\
    -0.2938 - 0.1456i & -0.0374 - 0.5898i & -0.6136 + 0.4081i\\
     0.5295 + 0.2625i & -0.0467 - 0.7355i &  0.2730 - 0.1816i\\
                   \end{array} \right]\left[\begin{array}{c}
        x_{3j+1} \\ x_{3j+2} \\ x_{3j+3} \\
       \end{array}
 \right],
\end{equation*}
with the symbols $x_i$, $i=1,\cdots,9$, taking values from a $4$-HEX constellation\footnote{A square $M$-HEX constellation of size $M$ is given by $\{a+\omega b ~\vert ~a,b \in \sqrt{M}\textrm{-PAM} \}$.}. The chosen unitary matrix $\mathbf{P}$, with reference to \eqref{eq_stbc_proposed}, is 
\begin{align*}
 \mathbf{P}=\begin{bmatrix}
             0 & -1 & 0\\
             1 & 0  & 0\\
             0 & 0 & 1
            \end{bmatrix}.
\end{align*}
The eigenvalues of $\mathbf{P}$ are $i$, $-i$ and $1$ which are distinct. Thus, from Theorem \ref{thm-max_sum_rate_Mx}, the transmission scheme for the above choice of $\mathbf{P}$ achieves the maximum sum-rate of 4 cspcu. For the low-delay transmission scheme, we employ the QPSK constellation rotated by an angle $\phi=\frac{\text{tan}^{-1}(2)}{2}$ which has a non-zero CPD \cite{ZaR}. From Theorem \ref{thm_x3_DoF}, the transmission scheme achieves the maximum sum-rate of 4 cspcu. We set $\theta=\frac{\pi}{4}$, and for this choice of $\theta$, a brute force computation for all pairs of difference matrices using the software MATLAB reveals that the proposed low-delay STBC is indeed full-ranked. The sum-rate achieved by both transmission schemes for the choice of their respective complex constellations is $8$ bits per channel use. The sphere decoder \cite{ViB} has been used to decode the transmitted symbols. The BER performances of both the transmission schemes are plotted in Fig. \ref{fig-BER_QPSK}. It can be inferred that the proposed schemes achieve a diversity gain of at least four which agrees with our analysis. 
\begin{figure}[htbp]
\centering
\includegraphics[totalheight=3.5in,width=5.8in]{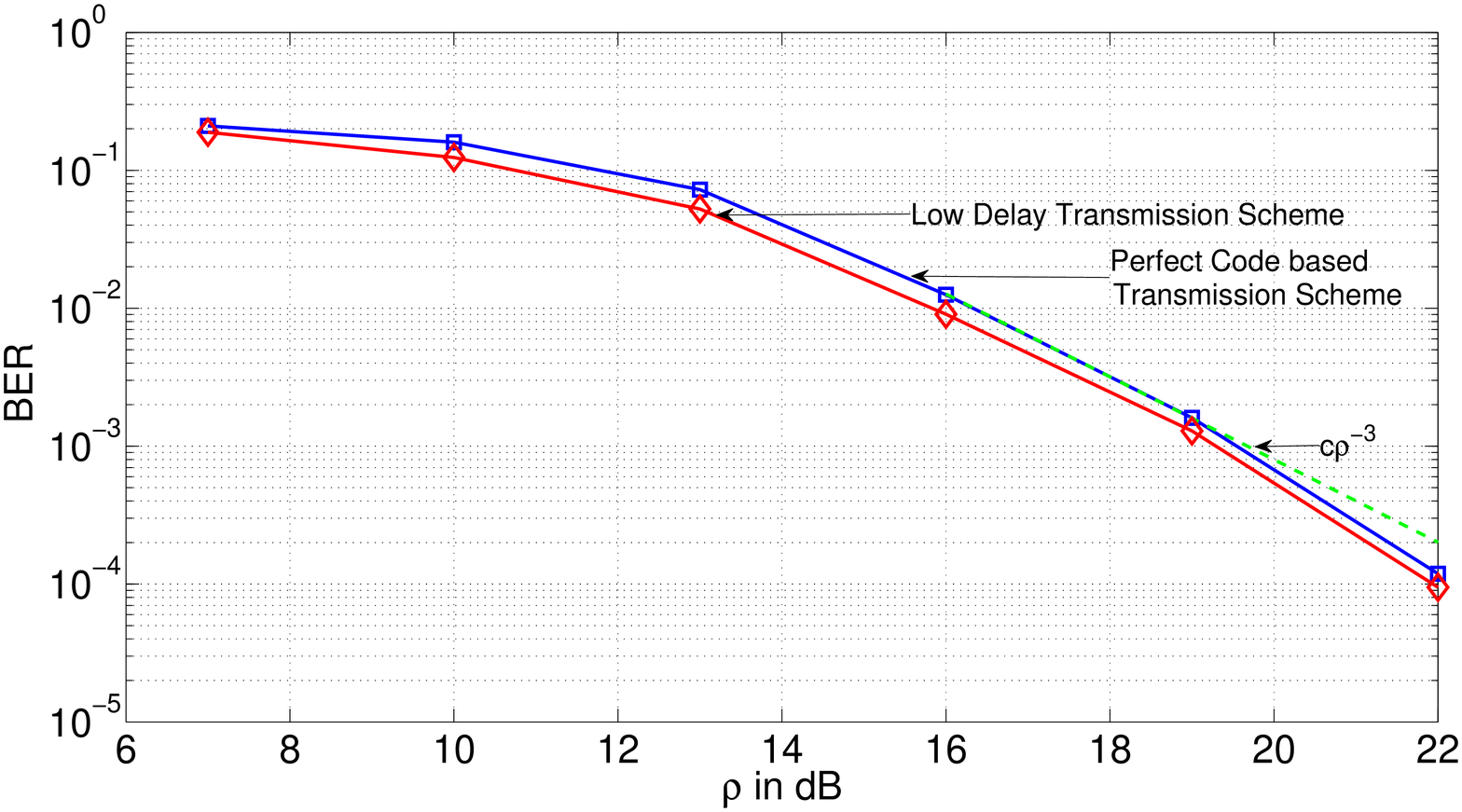}
\caption{BER comparison between the low delay transmission scheme and the Perfect STBC based transmission scheme at a sum-rate of $8$ bits per channel use using QPSK/HEX input constellations for the $(2\times 2, 3)$ X-Network. The low delay transmission scheme performs similar to the Perfect STBC based transmission scheme and also has lower decoding complexity because of lower delay. The dotted green line is plotted for some constant $c>0$.}
\label{fig-BER_QPSK}
\end{figure} 

\section{Concluding Remarks}\label{sec5}
In this paper, a maximum sum-rate transmission scheme for the $(2\times 2, M)$ X-Network was presented for arbitrary $M$. A new class of STBCs, namely STBCs with the column cancellation property, was introduced and used in the proposed transmission scheme. The proposed transmission scheme was shown to achieve the $\frac{4M}{3}$ sum-DoF of  the X-Network with only the availability of  local CSIT, whereas the Jafar-Shamai scheme \cite{JaS_X_Ch_2008} requires the availability of global CSIT in order to achieve the same. In addition, for block-fading channels, it was proven analytically that a diversity gain of $M+1$ is guaranteed when fixed finite input constellations are employed. Further, the known transmission schemes for the $(2\times 2, M)$ X-Network with $M=2,4$ \cite{LiJ,AbR_4x_X-Ch_TIT2014} were shown to be special cases of the transmission scheme proposed in this paper. 

With regards to diversity gain with fixed finite input constellations, it was shown that full receive diversity is achieved, but that the transmit diversity is affected due to channel-dependent precoding. While the achievability of a non-trivial diversity gain of $M+1$ was established for the $(2\times 2, M)$ X-Network, the intriguing possibility of achieving full transmit and full receive diversity at the maximum sum-rate transmission needs to be further investigated. This work also motivates the design of minimum-delay STBCs with the column cancellation property as a possible research direction.

\appendices
\section{Proof of Lemma \ref{lemma2}}\label{appen_lem_full_rank}

We first prove that for any matrix $\mathbf{P} = \textrm{diag}[\lambda_1,\lambda_2,\cdots,\lambda_M]$ with $\vert \lambda_i \vert^2 =1$, $i=1,\cdots,M$, there exists $\mathbf{A} \in \mathbb{C}^{M \times M}$ such that $\mathbf{AP}-\mathbf{PA}$ is full-ranked if and only if no more than $\left\lfloor\frac{M}{2} \right\rfloor$ of the $\lambda_i$ are equal. Denoting the $(i,j)^{th}$ entry of $\mathbf{A}$ by $a_{ij}$, $i,j = 1,2,\cdots,M$, we have
\begin{equation*}
 \mathbf{C}:=\mathbf{AP}-\mathbf{PA} = \left[ \begin{array}{ccccc}
                                 0 & (\lambda_2-\lambda_1)a_{12} & (\lambda_3-\lambda_1)a_{13} & \cdots & (\lambda_M-\lambda_1)a_{1M} \\
                                 (\lambda_1-\lambda_2)a_{21} & 0 & (\lambda_3-\lambda_2)a_{23}& \cdots & (\lambda_M-\lambda_2)a_{2M} \\
                                 (\lambda_1-\lambda_3)a_{31} & (\lambda_2-\lambda_3)a_{32} & 0 &\cdots & (\lambda_M-\lambda_3)a_{3M} \\
                                 \vdots & \vdots & \vdots & \ddots & \vdots \\
                                 (\lambda_1-\lambda_M)a_{M1}& (\lambda_2-\lambda_M)a_{M2} & (\lambda_3-\lambda_M)a_{M3}& \cdots & 0\\                                 
                                  \end{array}\right].
\end{equation*}

\noindent In short, the $(i,j)^{th}$ entry of $\mathbf{C} =\mathbf{AP}-\mathbf{PA}$ is $(\lambda_j-\lambda_i)a_{ij}$. Let $k_1,\cdots,k_l$ be the algebraic multiplicities of the eigenvalues of $\mathbf{P}$ with $k_1 \geq k_2\geq \cdots \geq k_l$ and $\sum_{i=1}^{l}k_i = M$. Without loss of generality, let $\lambda_{1} = \lambda_{2} = \cdots = \lambda_{k_1} \neq \lambda_{k_1+i}$, $i =1, \cdots,M-k_1$. Therefore, we have $\mathbf{C}_{11}:=\mathbf{C}(1:k_1,1:k_1) = \mathbf{0}$, and every entry (excepting the diagonal elements of $\mathbf{C}$) of $\mathbf{C}_{21}:=\mathbf{C}(k_1+1:M,1:k_1)$ is dependent on the choice of $a_{ij}$, $i=k_1+1,\cdots,M$, $j =1,\cdots,k_1$. Likewise, every entry of $\mathbf{C}_{12}:=\mathbf{C}(1:k_1,k_1+1:M)$ is dependent on the choice of $a_{ij}$, $i=1,\cdots,k_1$, $j =k_1+1,\cdots,M$. Since $\mathbf{C}_{11} = \mathbf{0}$, it is clear that for $\mathbf{C} =\mathbf{AP}-\mathbf{PA}$ to be full-ranked, $\mathbf{C}_{12} \in \mathbb{C}^{k_1\times (M-k_1)}$ and $\mathbf{C}_{21} \in \mathbb{C}^{(M-k_1)\times k_1}$ must be of rank $k_1$. But $Rank(\mathbf{C}_{12}),Rank(\mathbf{C}_{21}) \leq \min(k_1,M-k_1)$ and so, it must be that $k_1 \leq M-k_1$ so that $ k_1 \leq \left\lfloor\frac{M}{2} \right\rfloor$. 

To prove the converse, i.e., if $k_1 \leq \left\lfloor\frac{M}{2} \right\rfloor$, then there exists some assignment of values to $a_{ij}$ from $\mathbb{C}$ such that $\mathbf{C} =\mathbf{AP}-\mathbf{PA}$ is full-ranked, we make use of the following simple observation.
\begin{lemma}\label{lemma3}
 For a matrix $\mathbf{B} \in \mathbb{C}^{p \times p}$ that contains a null submatrix of size $m \times n$ and the remaining entries $b_{ij}$ are allowed to be chosen independently, there exists an assignment of values to $b_{ij}$ from $\mathbb{C}$ such that $Rank(\mathbf{B}) = p$ if $m+n \leq p$.  
\end{lemma}
\begin{proof}
Without loss of generality, let $\mathbf{B}(1:m,1:n) = \mathbf{0}_{m\times n}$ and let $n \leq m$. Now, $\mathbf{B}_1: = \mathbf{B}(m+1:p,1:p-m) \in \mathbb{C}^{(p-m)\times(p-m)}$ has entries all of which can be independently chosen from $\mathbb{C}$, and since $n \leq p-m$, the same holds true for $\mathbf{B}_2:=\mathbf{B}(1:m,p-m+1:p) \in \mathbb{C}^{m \times m}$. Choosing $\mathbf{B}_1$ and $\mathbf{B}_2$ to be full-ranked and $\mathbf{B}(m+1:p,p-m+1:p) = \mathbf{0}_{(p-m)\times m}$ ensures the full-rankness of $\mathbf{B}$.
\end{proof}

Since $k_1 \leq \left\lfloor\frac{M}{2} \right\rfloor$ and $\sum_{i=1}^L k_i=M$, there exists some $j < l$ such that $\sum_{i=1}^{j-1}k_i < \left\lfloor\frac{M}{2} \right\rfloor$ but $\sum_{i=1}^{j}k_i \geq \left\lfloor\frac{M}{2} \right\rfloor$ (with $k_0$ defined to be $0$). Consider the sub-matrices 
\begin{eqnarray*}
 \mathbf{C}_1& :=& \mathbf{C}\left(\left\lceil\frac{M}{2} \right\rceil:M,1:\left\lceil\frac{M}{2} \right\rceil \right)\in \mathbb{C}^{\left\lceil\frac{M}{2} \right\rceil \times \left\lceil\frac{M}{2} \right\rceil}, \\ \mathbf{C}_2& :=& \mathbf{C}\left(1:\left\lfloor\frac{M}{2} \right\rfloor,\left\lceil\frac{M}{2} \right\rceil+1:M \right) \in \mathbb{C}^{\left\lfloor\frac{M}{2} \right\rfloor \times \left\lfloor\frac{M}{2} \right\rfloor}
\end{eqnarray*}
 of $\mathbf{C} = \mathbf{AP}-\mathbf{PA}$. We assume without loss of generality that 
 \begin{align*}
&  \lambda_{1} = \lambda_{2} = \cdots = \lambda_{k_1} \neq \lambda_{k_1+i}, i =1,\cdots,M-k_1, \\
 & \lambda_{k_1+1} = \lambda_{k_1+2} = \cdots = \lambda_{k_1+k_2} \neq \lambda_{k_1+k_2+i}, i=1,\cdots,M-k_1-k_2, \\
  &\hspace{3cm}\vdots\\
  &\lambda_{k_1+k_2+\cdots+k_{l-1}+1} = \cdots = \lambda_{M}.
 \end{align*}
 
\begin{figure}[htbp] \vspace{-1.5cm}
\centering
\includegraphics[totalheight=5.2in,width=3.5in]{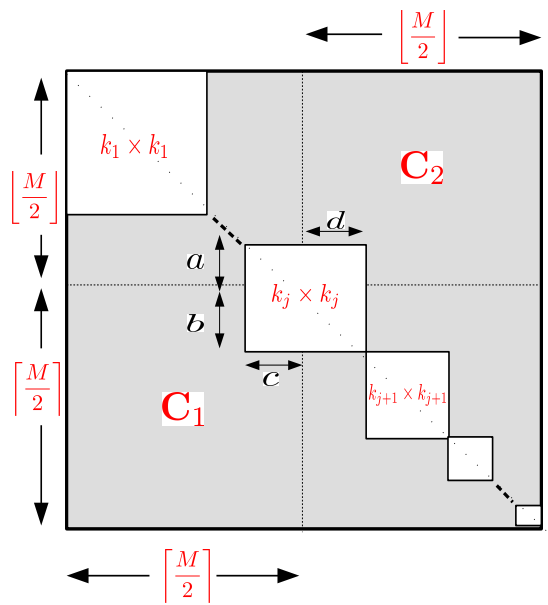}
\vspace{-2.5cm}
\caption{The structure of $\mathbf{AP}-\mathbf{PA}$ for a diagonal unitary matrix $\mathbf{P}$. The unshaded blocks denote the null matrices of size $k_i \times k_i$, $i=1,\cdots,l$, while the shaded region denotes the portion of $\mathbf{AP}-\mathbf{PA}$ which has entries that can be independently chosen from $\mathbb{C}$.}
\label{fig-matrix}
\end{figure} 

Therefore, $\mathbf{C}_1$ contains a null submatrix of size $b \times c$, with $b,c \leq k_j$, and the remaining entries of $\mathbf{C}_1$ are free to be chosen from $\mathbb{C}$ (see Fig. \ref{fig-matrix}). However, with $a = k_j-b$, we have that $c = a$ (if $M$ is even) and $c = a+1$ (if $M$ is odd). Therefore $b+c \leq k_j+1 \leq \left\lceil\frac{M}{2} \right\rceil$ because $k_{j-1} < \left\lfloor\frac{M}{2} \right\rfloor$ by assumption. Therefore, from Lemma \ref{lemma3}, $\mathbf{C}_1$ can be made full-ranked by a suitable choice of its non-zero entries. Following a similar argument, $\mathbf{C}_2$ too can be made non-singular by a suitable choice of its entries. Forcing $\mathbf{C}_1$ and $\mathbf{C}_2$ to be non-singular and the remaining entries of $\mathbf{C}$ to be zeros forces $\mathbf{C}$ to be non-singular as well. Hence, there does exist some assignment of values to $a_{ij}$ from $\mathbb{C}$ such that $\mathbf{C} =\mathbf{AP}-\mathbf{PA}$ is full-ranked. 

Now, for an arbitrary unitary matrix $\mathbf{P}$ that is not diagonal but has no eigenvalue with algebraic multiplicity exceeding $\left\lfloor\frac{M}{2} \right\rfloor$, we have $\mathbf{P} = \mathbf{UDU}^H$, obtained upon eigen-decomposition with $\mathbf{U}$ and $\mathbf{D}$ unitary, $\mathbf{D}$ diagonal. So, 
\begin{eqnarray*}
 \mathbf{AP} - \mathbf{PA} & = & \mathbf{AUDU}^H - \mathbf{UDU}^H\mathbf{A} \\
  & = & \mathbf{U}(\mathbf{U}^H\mathbf{AUD} - \mathbf{DU}^H\mathbf{AU})\mathbf{U}^H \\
  & = & \mathbf{U}(\mathbf{BD} - \mathbf{DB})\mathbf{U}^H
\end{eqnarray*}
where $\mathbf{B} := \mathbf{U}^H\mathbf{AU}$. So, $\mathbf{AP} - \mathbf{PA}$ is full-ranked if and only if $\mathbf{BD} - \mathbf{DB}$ also is. Applying the argument made in the previous paragraph, there exists $\mathbf{B} \in \mathbb{C}^{M \times M}$ for which $\mathbf{AP} - \mathbf{PA}$ is full-ranked. This proves Lemma \ref{lemma2}.

\section{Proof of Lemma \ref{lem_diff_mat_FR}}
\label{appen_lem_diff_mat_FR}
We prove that for every non-zero difference matrix, there exist at most a finite number of values of $\theta$ for which it is not full-ranked. Thus we conclude that there always exists $\theta$ such that all the non-zero difference matrices are full-ranked. 

Without loss of generality, we consider a difference matrix $\triangle {\bar{\mathbf{X}}} \neq \mathbf{0}$ which can be expressed as
\begin{align*}
   \triangle {\bar{\mathbf{X}}}= \left[\begin{array}{cccc}
        \triangle s_1 & e^{i\theta}\triangle s_4 & -\triangle s_2^* & -e^{i\theta}\triangle s_6^* \\
	\triangle s_2 &  e^{i\theta}\triangle s_5 &  \triangle s_1^* &e^{i\theta}\triangle s_4^* \\
	e^{i\theta}\triangle s_3  & \triangle s_6 & -e^{i\theta}\triangle s_3^* & -\triangle s_5^*\\
	\end{array}\right]
\end{align*}
where $ \triangle s_1 =  \triangle x_{1I}+i \triangle x_{3Q}$, $  \triangle s_2 =  \triangle x_{2I}+i  \triangle x_{4Q}$, $ \triangle s_3 =  \triangle x_{6I}+i \triangle x_{5Q}$, $ \triangle s_4 =  \triangle x_{5I}+i \triangle x_{6Q}$, $ \triangle s_5 =  \triangle x_{4I}+i \triangle x_{2Q} $, $ \triangle s_6 =  \triangle x_{3I}+i \triangle x_{1Q}$, with $ \triangle x_i$, $i=1,2,\cdots,6$, being the difference symbols.
Consider the matrices $\mathbf{A},\mathbf{B} \in \mathbb C^{3\times 3}$ comprised of the first three columns and the last three columns of $\triangle {\bar{\mathbf{X}}}$ respectively. Expanding along the second column, the determinant of $\mathbf{A}$ is 
\begin{eqnarray}
\nonumber
-det(\mathbf{A})&=&e^{2i\theta}\triangle s_4\left(-\triangle s_2 \triangle s_3^* -\triangle s_1^*\triangle s_3\right)
-e^{2i\theta}\triangle s_5\left(-\triangle s_1 \triangle s_3^* +\triangle s_2^*\triangle s_3\right)\\
\label{eqn-det_A}
&&+\triangle s_6\left(\vert \triangle s_1 \vert^2 +\vert \triangle s_2 \vert^2\right).
\end{eqnarray}
Expanding along the second column of $\mathbf{B}$, its determinant is
\begin{eqnarray}
 \nonumber
-det(\mathbf{B})&=&e^{i\theta}\triangle s_2^*\left(-\vert \triangle s_5 \vert^2 -\triangle s_4^*\triangle s_6\right)
+e^{i\theta}\triangle s_1^*\left(-\triangle s_4 \triangle s_5^* +\vert \triangle s_6\vert^2\right)\\
\label{eqn-det_B}
&&+e^{2i\theta}\triangle s_3^*\left(\vert \triangle s_4\vert^2 +\triangle s_5\triangle s_6^*\right).
\end{eqnarray}

\textbf{Case 1:} Consider the case {$\left(\triangle x_{1I},\triangle x_{3I}\right)=(0,0)$} and {$\left(\triangle x_{5I},\triangle x_{6I}\right)=(0,0)$}. Here, the determinant of $\mathbf{B}$ is 
\begin{align*}
  det(\mathbf{B})=e^{i\theta}\triangle s_2^*\vert \triangle s_5 \vert^2. 
\end{align*} Since $\triangle {\bar{\mathbf{X}}} \neq \mathbf{0}$, either $\triangle x_{2I}$ or $\triangle x_{4Q}$ or both of them are non-zero. Hence, $  det(\mathbf{B})\neq 0$ and $\triangle {\bar{\mathbf{X}}}$ is of rank $3$.

\textbf{Case 2:} Consider the case {$\left(\triangle x_{1I},\triangle x_{3I}\right)\neq(0,0)$} and {$\left(\triangle x_{5I},\triangle x_{6I}\right)=(0,0)$}. The determinant of $\mathbf{A}$ is given by
{\begin{align*}
det(\mathbf{A})=-\triangle s_6\left(\vert \triangle s_1 \vert^2 +\vert \triangle s_2 \vert^2\right).
\end{align*}}Since $\triangle x_{3I}$ or $\triangle x_{1Q}$ or both are non-zero, $det(\mathbf{A}) \neq 0$ for this case. Hence,  $\triangle {\bar{\mathbf{X}}}$ is of rank $3$.

\textbf{Case 3:} Consider the case {$\left(\triangle x_{1I},\triangle x_{3I}\right)=(0,0)$} and {$\left(\triangle x_{5I},\triangle x_{6I}\right)\neq(0,0)$}. In this case, the coefficient of $e^{2i\theta}$ in the determinant of the matrix $\mathbf{B}$ is given by {$\left(-\triangle s_3^*\right)\left({\triangle x_{5I}}^2+{\triangle x_{6Q}}^2\right) \neq 0$}. Now, $  det(\mathbf{B})$ is a quadratic polynomial in $e^{i\theta}$ which can have at most two roots for $e^{i\theta}$, and hence at most a finite number of values of $\theta$ for which $  det(\mathbf{B})=0$. Therefore, there exist infinite values of $\theta$ for which $  det(\mathbf{B}) \neq 0$ in this case.

\textbf{Case 4:} Consider the case {$\left(\triangle x_{1I},\triangle x_{3I}\right)\neq(0,0)$} and {$\left(\triangle x_{5I},\triangle x_{6I}\right)\neq(0,0)$}. If the first two terms of $det(\mathbf{A})$ given in (\ref{eqn-det_A}) do not sum to zero then, $det(\mathbf{A})$ is clearly a quadratic polynomial in $e^{i\theta}$. Thus, there exist infinite values of $\theta$ for which $det(\mathbf{A})$ is non-zero. If the first two terms of $det(\mathbf{A})$ sum to zero then,  $det(\mathbf{A}) \neq 0$ for the same reason as in Case $2$. Hence, $\triangle {\bar{\mathbf{X}}}$ is of rank $3$ in this case also.

\section{Proof of Theorem \ref{thm_x3_DoF}}
\label{appen_thm_x3_DoF}
Referring to \eqref{eqn-rx1_IF}, the interference-free processed received symbol matrix $\mathbf{Y}'_1 \in \mathbf{C}^{3\times 4}$ is given by 
\begin{align*} 
&\mathbf{Y}'_1=\sqrt{\frac{3\rho}{4}}\left(\mathbf{H}_{11}\mathbf{V}_{11}\bar{\mathbf{X}}_{11}+\mathbf{H}_{21}\mathbf{V}_{21}\bar{\mathbf{X}}_{21}\right)+\mathbf{N}'_1,
\end{align*}
where $\mathbf{N}'$ is a noise matrix whose entries are independent. We have $\mathbf{N}'(i) \sim {\cal CN}(0,\mathbf{I}_3)$, $i=1,2$, and $\mathbf{N}'(i) \sim {\cal CN}(0,2\mathbf{I}_3)$, $i=3,4$. Since increasing the noise variance affects neither the achieved DoF nor the diversity gain, we assume that $\mathbf{N}'(i) \sim {\cal CN}(0,2\mathbf{I}_3)$, $i=1,2,3,4$. The matrices $\bar{\mathbf{X}}_{i1}$ have the structure given in \eqref{eqn-M=3_STBC}. Specifically, 
\begin{align*}
   \bar{\mathbf{X}}_{i1}= \left[\begin{array}{cccc}
     s_{1}^{(i1)} & e^{i\theta}s_{4}^{(i1)} & -\left(s_{2}^{(i1)}\right)^* &   -e^{i\theta}\left(s_{6}^{(i1)}\right)^* \\
	s_{2}^{(i1)} & e^{i\theta}s_{5}^{(i1)} & \left(s_{1}^{(i1)}\right)^* &  e^{i\theta}\left(s_{4}^{(i1)}\right)^* \\
	e^{i\theta}s_{3}^{(i1)} & s_{6}^{(i1)} & -e^{i\theta}\left(s_{3}^{(i1)}\right)^*& -\left(s_{5}^{(i1)}\right)^* \\
     \end{array}\right]
\end{align*}

\noindent where $s_{1}^{(i1)} = x_{1I}^{(i1)}+ix_{3Q}^{(i1)}$, $s_{2}^{(i1)} = x_{2I}^{(i1)}+ix_{4Q}^{(i1)}$, $s_{3}^{(i1)} = x_{6I}^{(i1)}+ix_{5Q}^{(i1)}$, $s_{4}^{(i1)} = x_{5I}^{(i1)}+ix_{6Q}^{(i1)}$, $s_{5}^{(i1)} = x_{4I}^{(i1)}+ix_{2Q}^{(i1)} $, $s_{6}^{(i1)} = x_{3I}^{(i1)}+ix_{1Q}^{(i1)}$, with $x_{j}^{(i1)}$, $i=1,2$, $j=1,\cdots,6$, taking values from a suitable complex constellation. Let $\mathbf{x}_{i1} := [x_{1}^{(i1)} ~ x_{2}^{(i1)}~ x_{3}^{(i1)}~x_{4}^{(i1)}~x_{5}^{(i1)}~x_{6}^{(i1)} ]^T $, and \[\mathbf{s}_{i1}: = \left[s_{1}^{(i1)} ~ s_{2}^{(i1)} ~ e^{i\theta}s_{3}^{(i1)} ~ e^{i\theta}s_4^{(i1)} ~ e^{i\theta}s_5^{(i1)} ~ s_6^{(i1)} \right]^T = \left[\bar{\mathbf{X}}_{i1}(1)^T ~ \bar{\mathbf{X}}_{i1}(2)^T\right]^T. \] It is evident that $\mathbf{x}_{i1}$ can be completely recovered from $\mathbf{s}_{i1}$. We therefore have
\[ \bar{\mathbf{X}}_{i1} =\left[\bar{\mathbf{X}}_{i1}(1) ~ \bar{\mathbf{X}}_{i1}(2) ~ \mathbf{P}'_1\left(\bar{\mathbf{X}}_{i1}(1)\right)^* ~ \mathbf{P}'_2\left(\bar{\mathbf{X}}_{i1}(2)\right)^* \right]\] with 
\begin{equation*}
 \mathbf{P}'_1 = \left[ \begin{array}{ccc}
                        0 & -1 &  0 \\
                        1 & 0 & 0 \\
                        0 & 0 & -e^{2i\theta}\\
                        \end{array}\right], ~~~ \mathbf{P}'_2 = \left[ \begin{array}{ccc}
                        0 & 0 & -e^{i\theta}\\
                        e^{2i\theta} & 0 & 0\\                       
                        0 & -e^{i\theta} & 0 \\
                       \end{array}\right],
\end{equation*}

Therefore, \begin{align*}
		 \left[\begin{array}{c}
                       \mathbf{Y}_1(1)\\
                       \mathbf{Y}_1(2)\\  
                       (\mathbf{Y}_1(3))^*\\
                       (\mathbf{Y}_1(4))^*\\
                      \end{array}\right] =  \sqrt{\frac{3\rho}{4}}\left[\begin{array}{cccc}
                     \mathbf{H}_1 & \mathbf{0} &  \mathbf{H}_2   & \mathbf{0}\\
                      \mathbf{0} &  \mathbf{H}_1 & \mathbf{0} & \mathbf{H}_2 \\
                      \left(\mathbf{H}_1\mathbf{P}'_1\right)^* & \mathbf{0} & \left(\mathbf{H}_2\mathbf{P}'_2\right)^* &  \mathbf{0}\\
                      \mathbf{0} & \left(\mathbf{H}_1\mathbf{P}'_1\right)^*  & \mathbf{0} &  \left(\mathbf{H}_2\mathbf{P}'_2\right)^* \\                     
                      \end{array}                      
                      \right]\left[\begin{array}{c}
                       \mathbf{s}_{11}\\
                       \mathbf{s}_{21}\\                       
                      \end{array}\right]+ \left[\begin{array}{c}
                       \mathbf{n}'\\
                      \end{array}\right],
\end{align*}
where $\mathbf{H}_{11}\mathbf{V}_{11} := \mathbf{H}_1$, $ \mathbf{H}_{21}\mathbf{V}_{21} := \mathbf{H}_2$, and $\mathbf{n}' \sim \mathcal{CN}(\mathbf{0},2\mathbf{I}_{12})$. To establish that Rx-$1$ receives 2 linearly independent complex symbols per channel use (i.e., 12 linearly independent complex symbols in 6 channel uses), it is sufficient to prove that the matrix 
\begin{equation*}
 \mathbf{H}' := \left[\begin{array}{cc}
                       \mathbf{H}_1 &  \mathbf{H}_2 \\
                      \left(\mathbf{H}_1\mathbf{P}'_1\right)^* & \left(\mathbf{H}_2\mathbf{P}'_2\right)^*\\
                      \end{array} \right]
\end{equation*}
is full-ranked almost surely. Let us assign $\mathbf{H}_1 = \mathbf{I}_3$. To prove that $det(\mathbf{H}')$ is not identically $0$, it is sufficient to prove that $det(\mathbf{H}')$ is a non-zero polynomial in the rest of the variables, with $\mathbf{H}_1 = \mathbf{I}_3$. Now, the determinant $det(\mathbf{H}') = 0$ iff $det\left(\left(\mathbf{H}_2\mathbf{P}'_2\right)^*-\mathbf{P}'^*_1\mathbf{H}_2\right)$ is a zero polynomial. But we show that, for any $\theta \in [0,2\pi)$, there exists an assignment to the channel matrices so that $\left(\mathbf{H}_2\mathbf{P}'_2\right)^*-\mathbf{P}'^*_1\mathbf{H}_2$ is not singular. Let

\begin{equation*}
 \mathbf{H}_2 = \left[ \begin{array}{ccc}
                        0 & 0 & -e^{-2i\theta} \\
                        0 & 2 & 0 \\
                        1 & 0 & 0 \\
                       \end{array}\right].
\end{equation*}So, we have
\begin{equation*}
 det(\left(\mathbf{H}_2\mathbf{P}'_2\right)^*-\mathbf{P}'^*_1\mathbf{H}_2) = e^{-3i\theta}\left(2+e^{-i\theta}\right)\left(2+e^{i\theta}\right) \neq 0, \forall ~\theta \in [0,2\pi).
\end{equation*}Therefore, choosing $\theta$ to be anything in $[0,2\pi)$ ensures that $ \mathbf{H}'$ is full-ranked with probability 1, and this completes the proof of Theorem \ref{thm_x3_DoF}.  

\section*{Acknowledgement}
This work was partially supported by a grant from University Grants Committee of the Hong Kong Special Administrative Region, China (Project No. AoE/E-02/08). The authors would like to thank Prof. B. Sundar Rajan for his involvement in the preliminary version of a part of this work \cite{AbR_X_Ch_ISIT2014}.

%
%
 \bibliographystyle{ieeetr}
 \bibliography{References}


\end{document}